\newif\ifSHORTversion
\SHORTversionfalse	
\ifSHORTversion
  \documentclass[12pt,draftcls,onecolumn]{IEEEtran}
\else
  \documentclass[draftcls,onecolumn,journal]{IEEEtran}
\fi

\usepackage{ifthen}

\pdfoutput=1

\ifCLASSINFOpdf
   \usepackage[pdftex]{graphicx}
   \graphicspath{{./pdfFig/}{../jpegFig/}{./pdfFig/sim/}}
   \DeclareGraphicsExtensions{.pdf,.jpeg,.png}
\else
\fi

\usepackage{cite}
\usepackage{graphics} 
\usepackage{graphicx}

\usepackage{tikz}
\usetikzlibrary{shapes}

\usepackage{epstopdf}
\usepackage{amsmath} 
\usepackage{amssymb}  
\usepackage{subfigure}
\usepackage{xcolor}
\usepackage{algorithm}
\usepackage{algorithmic}

\usepackage{pdfsync}

\usepackage{makecmds}

\newcommand{\emath}[1]{\ensuremath{#1}}

\newcommand{\ivar}{\emath{i}}
\newcommand{\jvar}{\emath{j}}

\newcommand{\kvar}{\emath{k}}

\newcommand{\xvar}{\emath{x}}
\newcommand{\yvar}{\emath{y}}

\newcommand{\probNotHold}{q}

\newcommand{\probability}{\emath{\mathbb{P}}} 
\newcommand{\prob}[1]{ \mbox{$\mathbb{P}[#1]$} }
\newcommand{\probcond}[2]{\mbox{$\mathbb{P}[#1| \,#2]$}}
\newcommand{\expectation}[1]{\mbox{$\mathbb{E}\left[#1\right]$}}
\newcommand{\condexpectation}[2]{\mbox{$\mathbb{E}\left[#1\,|\, #2\right]$}} 

\newcommand{\expect}{\emath{\mathbb{E}}}
\newcommand{\pdistr}{\emath{\varphi}}

\ifdefined\kylehackusingbeamer
\else
	
	\newtheorem{theorem}{Theorem}[section]
	
	\newtheorem{problem}{Problem}
	\newtheorem{lemma}[theorem]{Lemma}
	
	\newtheorem{remark}[theorem]{Remark}
\fi
	\newtheorem{prop}[theorem]{Proposition}

\provideenvironment{example}[1][Example]{
\begin{trivlist}
\item[\hskip \labelsep {\bfseries #1}]
\it	
}{ \end{trivlist} }

\newcommand{\reals}{ \emath{\mathbb{R}} }

\newcommand{\vecnorm}[1]{\emath{\left\|#1\right\|}}

\newcommand{\probtext}{\operatorname*{prob}}
\newcommand{\xpose}{\emath{ {\rm T} }}

\newcommand{\locvar}{\emath{x}}
\newcommand{\locvartwo}{\emath{y}}

\newcommand{\picktag}{\emath{ {\rm P} }}
\newcommand{\delvtag}{\emath{ {\rm D} }}

\newcommand{\env}{\emath{\Omega}}

\newcommand{\diam}{\mathrm{D}(\env)}

\newcommand{\zb}{\mathcal Z}

\newcommand{\arrivalrate}{\emath{\lambda}}
\newcommand{\utilization}{\emath{\varrho}}

\newcommand{\numveh}{\emath{m}}

\newcommand{\den}{\varphi}

\newcommand{\avglink}{ \travelbetween }

\newcommand{\algog}{\mathcal M}
\newcommand{\indicator}[2]{ { I_{#1}{ \left(#2\right) } } }

\newcommand{\numdem}{\emath{n}}

\newcommand{\oneover}[1]{\emath{\frac{1}{#1}}}

\newcommand{\travelbetween}{ \emath{l} }

\newcommand{\probdim}{\emath{d}}
\newcommand{\tourlen}{\emath{L}}

\newcommand{\sidelength}{\emath{L}}
\newcommand{\sidecuts}{\emath{r}}
\newcommand{\workcell}{\emath{C}}
\newcommand{\cellset}{\emath{\mathcal{C}}}
\newcommand{\transport}{\emath{A}}

\newcommand{\Xvar}{\emath{X}}
\newcommand{\Yvar}{\emath{Y}}

\newcommand{\subn}{\emath{N}}

\newcommand{\cyclevar}{\mathcal L}

\newcommand{\permcyclecount}{N}
\newcommand{\batchsize}{\emath{n}}
\newcommand{\batchset}{\emath{\mathcal S}}
\newcommand{\batchvar}{\emath{s}}
\newcommand{\permset}{\emath{\Pi}}

\newcommand{\matchperm}{\emath{ {\permset^*} }}
\newcommand{\permvar}{\emath{\sigma}}

\newcommand{\M}{L_{\text{M}}}

\newcommand{\eps}{\varepsilon}

\newcommand{\pickset}{{\mathcal X}}
\newcommand{\delvset}{{\mathcal Y}}
\newcommand{\ESCPinst}{ {\text{ESCP}(\numdem,\den_\picktag,\den_\delvtag)} }
\newcommand{\transportset}{ {\mathcal T} }
\newcommand{\spatialjoint}{{\den_\picktag\den_\delvtag}}

\newcommand{\permelarg}[1]{ {\permvar\left({#1}\right)} }
\newcommand{\permelem}[1]{ { \permelarg{#1} } }

\newcommand{\batchsetdesc}{\env^{2n}}

\definecolor{gold}{rgb}{.75,.54,.25}

\provideenvironment{tourmaker}[1][.6]{
\begin{tikzpicture}[scale=#1]
\tikzstyle{pickup}=[circle,draw=red!100] 
\tikzstyle{delivr}=[circle,dashed,draw=blue!100] 
\tikzstyle{p2d}=[->,thick]
\tikzstyle{stacker}=[->,thick,dashed,color=olive]
\tikzstyle{match}=[->,thick,dashed,color=blue]
\tikzstyle{primary}=[color=green,->,thick,dashed]
\tikzstyle{node-done} = [circle,draw=gray!100] 
\tikzstyle{edge-done} = [color=gray,->,thick]
\coordinate (p1) at (-3.5,.5) ;		\coordinate (d1) at (-2.5, 3.5) ;	
\coordinate (p2) at (-.5,2.5) ;		\coordinate (d2) at (-3, -2) ;		
\coordinate (p3) at (-1.5,-2.5) ;	\coordinate (d3) at (-.75, 0) ;		
\coordinate (p4) at (1.5,-3) ;		\coordinate (d4) at (3, -1) ;		
\coordinate (p5) at (3.5,.5) ;		\coordinate (d5) at (2.5, 2.5) ;	
\coordinate (p6) at (1.5,.5) ;		\coordinate (d6) at (.5, -1.5) ;
\newcommand{\instsize}{6}
\newcommand{\instmatching}{ 1/2, 2/3, 3/1, 4/5, 5/6, 6/4 }
\newcommand{\randperm}{ 1/2, 2/5, 3/1, 4/4, 5/6, 6/3 }
\newcommand{\stackerlinks}{ 1/2, 2/3, 4/5, 6/4, 3/6 }
\newcommand{\stackerplus}{ 1/2, 2/3, 4/5, 6/4, 3/6, 5/1 }

\draw[step=1,color=blue!50,very thin] (-4,-4) grid (4,4) ;		
\newcommand{\drawnodes}{
\foreach \x in { 1,...,\instsize }{
  \node [delivr] (delv\number\x) at (d\number\x) {$\x$} ;
  \node [pickup] (pick\number\x) at (p\number\x) {$\x$} ;
} }
\newcommand{\drawlinks}{
\foreach \x in { 1,...,\instsize }{
  \draw [p2d] (pick\number\x) -- (delv\number\x) ;
} }
\newcommand{\drawpnodes}{
\foreach \x / \y in \randperm {
  \node [delivr] (delv\number\y) at (d\number\y) {$\y$} ;
  \node [pickup] (pick\number\x) at (p\number\x) {$\y$} ;
} }
\newcommand{\drawplinks}{
\foreach \x / \y in \randperm {
  \draw [p2d] (pick\number\x) -- (delv\number\y) ;
} }
\newcommand{\drawmatching}{
\foreach \x / \y in { 1/2, 2/3, 3/1, 4/5, 5/6, 6/4 }{
  \draw [match] (delv\number\x) -- (pick\number\y) ;
} }
\newcommand{\drawstacker}{
\foreach \x / \y in \stackerlinks {
  \draw [stacker] (delv\number\x) -- (pick\number\y) ;
} }
\newcommand{\drawstackerplus}{
\foreach \x / \y in \stackerplus {
  \draw [stacker] (delv\number\x) -- (pick\number\y) ;
} }
}{ \end{tikzpicture} }

\newcommand{\concat}{\operatorname{concat}}

\begin{document}

\title{Asymptotically Optimal Algorithms for Pickup and Delivery Problems with Application to Large-Scale Transportation Systems}

\author{Kyle Treleaven, Marco Pavone, Emilio Frazzoli
\thanks{Kyle Treleaven and Emilio Frazzoli are with the Laboratory for Information and Decision Systems, Department of Aeronautics and
Astronautics, Massachusetts Institute of Technology, Cambridge, MA 02139 {\tt\small \{ktreleav, frazzoli\}@mit.edu}.}%
\thanks{Marco Pavone is with the Department of Aeronautics and Astronautics, Stanford University, Stanford, CA 94305  {\tt\small pavone@stanford.edu}.}
\thanks{This research was supported in part by the Future Urban Mobility project of the Singapore-MIT Alliance for Research and Technology (SMART) Center, with funding from Singapore's National Research Foundation. }}
\maketitle

\begin{abstract}
Pickup and delivery problems (PDPs), in which objects or people have to be transported between specific locations, are among the most common combinatorial problems in
real-world logistical operations.
A widely-encountered type of PDP is the Stacker Crane Problem (SCP), where each commodity/customer is associated with a pickup location and a delivery location,
and the objective is to find a minimum-length tour visiting all locations with the constraint that each pickup location and its associated delivery location are visited in consecutive order.
The SCP is NP-Hard and the best known approximation algorithm only provides a 9/5 approximation ratio.
The objective of this paper is threefold. First, by embedding the problem within a stochastic framework, we present a novel algorithm for the SCP that:
(i) is asymptotically optimal, i.e., it produces, almost surely, a solution approaching the optimal one as the number of pickups/deliveries goes to infinity;
and (ii) has computational complexity $O(n^{2+\eps})$, where $n$ is the number of pickup/delivery pairs and $\eps$ is an arbitrarily small positive constant.
Second, we asymptotically characterize the length of the optimal SCP tour.
Finally, we study a dynamic version of the SCP, whereby pickup and delivery requests arrive according to a Poisson process, and which serves as a model for large-scale demand-responsive transport (DRT) systems. For such  a dynamic counterpart of the SCP, we derive a necessary and sufficient condition for the existence of stable vehicle routing policies, which depends only on the workspace
geometry, the stochastic distributions of pickup and delivery points, the arrival rate of requests, and the number of vehicles.
Our results leverage a novel connection between the Euclidean Bipartite Matching Problem and the theory of random permutations, and, for the dynamic setting, exhibit novel
features that are absent in traditional spatially-distributed queueing systems.

\end{abstract}

\section{Introduction}

Pickup and delivery problems (PDPs) constitute an important class of vehicle routing problems in which objects or people have to be transported between locations in a physical environment.
These problems arise in many contexts such as logistics, transportation systems, and robotics, among others. Broadly speaking, PDPs can be divided into three classes \cite{Berbeglia.Cordeau.ea:TOP07}: 1) Many-to-many PDPs, characterized by several origins and destinations for each commodity/customer; 2) one-to-many-to-one PDPs, where commodities are initially available at a depot and are destined to customers' sites, and commodities available at customers' sites are destined to the depot (this is the typical case for the collection of empty cans and bottles);
and 3) one-to-one PDPs, where each commodity/customer has a given origin and a given destination. 

When one adds capacity constraints to transportation vehicles and disallows transshipments, the one-to-one PDP is commonly referred to as the Stacker Crane Problem (SCP). The SCP is a route optimization problem at the core of several transportation systems, including demand-responsive transport (DRT) systems, where users formulate requests for transportation from a pickup point to a delivery point \cite{CJF-GL-ea:07, WJM-CEB-LDB:10}. Despite its importance, current algorithms for its solution are either of exponential complexity or come with quite poor guarantees on their performance; furthermore, most of the literature on the SCP does not consider the dynamic setting where pickups/deliveries are revealed sequentially in time. 
Broadly speaking, the objective of this paper is to devise polynomial-time algorithms for the SCP with \emph{probabilistic optimality guarantees}, and derive stability conditions for its \emph{dynamic} counterpart (where pickup/delivery requests are generated by an exogenous Poisson process and that serves as a model for DRT systems).

\emph{Literature overview}. The SCP, being a generalization of the Traveling Salesman Problem, is NP-Hard \cite{GNF.DJG:JC93}. The problem is NP-Hard even on trees, since the Steiner Tree Problem can be reduced to it \cite{FG:93}. In \cite{FG:93}, the authors present several approximation algorithms for tree graphs with a worst-case performance ratio ranging from 1.5 to around 1.21. The 1.5 worst-case algorithm, based on a Steiner tree approximation, runs in linear time. Recently, one of the polynomial-time algorithms presented in \cite{FG:93} has been shown to provide an optimal solution on almost all inputs (with a 4/3-approximation in the worst case) \cite{Krumke:JA06}. Even though the problem is NP-hard on general trees, the problem is in P on paths \cite{atallah}. For general graphs, the best approximation ratio is 9/5 and is achieved by an algorithm in \cite{FHK:S76}. Finally, an average case analysis of the SCP on trees has been examined for the special case of caterpillars as underlying graphs \cite{CO:P03}.

Dynamic SCPs are generally referred to in the literature as dynamic PDPs with unit-capacity vehicles (1-DPDPs); in the dynamic setting pickup/delivery requests are generated by an exogenous Poisson process and the objective is to minimize the waiting times of the requests. 1-DPDPs represent effective models for one-way vehicle sharing systems, which constitute a leading paradigm for future urban mobility \cite{WJM-CEB-LDB:10}. They are generally treated as a sequence of static subproblems and their performance properties, such as stability conditions, are, in general, not characterized analytically.
Thorough surveys on heuristics, metaheuristics  and online algorithms for 1-DPDPs can be found in \cite{GB-JFC-GL:10} and \cite{SNP-KFD-RFH:08}. Even though these algorithms are quite effective in addressing 1-DPDPs, alone they do not give any indication of fundamental limits of performance. To the best of our knowledge, the only \emph{analytical} studies for 1-DPDPs are \cite{Swihart.Papastavrou:99} and \cite{Treleaven.Pavone.ea:10}. Specifically, in \cite{Swihart.Papastavrou:99} the authors study the single vehicle case of the problem under the constraint that pickup and delivery distributions are
uniform; in \cite{Treleaven.Pavone.ea:10} the authors derive bounds for the more general case of multiple vehicles and general distributions, however under the quite unrealistic assumption of three-dimensional workspaces and identical distributions of pickup and delivery sites.

\emph{Contributions}. In this paper, we embed the SCP within a probability framework where origin/destination pairs are identically and independently distributed random variables within an Euclidean environment. Our random model is general in the sense that we consider potentially non-uniform distributions of points, with an emphasis on the case that the distribution of pickup sites is distinct from that of delivery sites; furthermore, the graph induced by the origin/destination pairs does not have any specific restrictions. We refer to this version of the SCP as the \emph{stochastic}  SCP.

The contribution of this paper is threefold. First, we present the \ref{alg:splicing} algorithm,
a polynomial-time algorithm for the stochastic SCP,
which is asymptotically optimal almost surely;
that is, except on a set of instances of zero probability,
the cost of the tour produced by this algorithm approaches the optimal cost as the number $\batchsize$ of origin/destination pairs goes to infinity.
In practice, convergence is very rapid and \ref{alg:splicing} computes solutions for the SCP within $5\%$ of the optimal cost for a number of pickup/delivery pairs as low as $100$.
The \ref{alg:splicing} algorithm has complexity of the order $O(\batchsize^{2+\epsilon})$ (for arbitrarily small positive constant $\epsilon$), where $\batchsize$ is the number of pickup/delivery pairs.
From a technical standpoint, these results leverage a novel connection between the Euclidean Bipartite Matching Problem and the theory of random permutations.

Second, we provide bounds for the cost of the optimal SCP solution, and also for the solution delivered by \ref{alg:splicing} (these bounds are asymptotic and hold almost surely).

Finally, by leveraging the previous results, we derive a necessary and sufficient stability condition for 1-DPDPs for the general case of multiple vehicles and possibly \emph{different} distributions for pickup and delivery sites. We show that when such distributions are different, our stability condition presents an additional (somewhat surprising) term compared to stability conditions for traditional spatially-distributed queueing systems.
This stability condition depends only on the workspace geometry, the stochastic distributions of pickup and delivery points, the demand arrival rate, and the number of vehicles.

\emph{Structure of the paper}.
This paper is structured as follows.
In Section~\ref{sec:background} we provide some background on the Euclidean Bipartite Matching Problem and on some notions in probability theory and transportation theory.
In Section~\ref{sec:problem} we rigorously state the stochastic SCP, the 1-DPDP, and the objectives of the paper;
in Section~\ref{sec:scpoptimal} we introduce and analyze the \ref{alg:splicing} algorithm, a polynomial-time, asymptotically optimal algorithm for the stochastic SCP.
In Section~\ref{sec:scpbounds} we derive a set of analytical bounds on the cost of a stochastic SCP tour,
and in Section~\ref{sec:dynamicversion} we use our results to obtain a general necessary and sufficient condition for the existence of stable routing policies for 1-DPDPs. Then, in Section~\ref{sec:sim} we present simulation results
corroborating our findings. Finally, in Section~\ref{sec:conclusion}, we draw some conclusions and discuss some directions for future work.

\section{Background Material}
\label{sec:background}

In this section we summarize the background material used in the paper.
Specifically, we review some results in permutation theory, the stochastic Euclidean Bipartite Matching Problem (EBMP), a related concept in transportation theory, and a generalized Law of Large Numbers.

\subsection{Permutations}\label{subsec:cycle}

A permutation is a rearrangement of the elements of an ordered set $\batchset$ according to a \emph{bijective} correspondence $\permvar :\batchset \to \batchset$.
As an example, a particular permutation of the
set $(1,2,3,4)$
is $\permvar(1) = 3$, $\permvar(2) = 1$, $\permvar(3) = 2$, and $\permvar(4) = 4$, which leads to the reordered set $(3,1,2,4)$.
The number of distinct permutations on a set of $n$ elements is given by $n!$ (factorial).
We denote the set of permutations over the $n$-element ordered set $(1,\ldots,n)$ by $\permset_n$.
A permutation can be conveniently represented in a two-line notation, where one lists the elements of $\batchset$ in the first row and their images in the second row, with the property that a first-row element and its image are in the same column.
For the previous example, one would write:
\begin{equation}\label{eq:experm}
\left[ \begin{array}{cccc}
1 & 2& 3 & 4\\
3 & 1& 2 & 4
 \end{array} \right].
\end{equation}
The \emph{identity permutation} maps every element of a set $\batchset$ to itself and will be denoted by $\permvar_1$.
We will use the following elementary properties of permutations, which follow from the fact that permutations are bijective correspondences:
\begin{description}
\item[Prop.] 1: Given two permutations $\permvar, \, \hat  \permvar \in \permset_n$, the composition
$\permvar \hat \permvar$ is again a permutation.
\item[Prop.] 2: Each permutation $\permvar \in \permset_n$ has an inverse permutation $\permvar^{-1}$, with the property that $\permvar(x) = y$ if and only if $\permvar^{-1}(y)= x$. (Thus, note that $\permvar^{-1} \permvar = \permvar_1$.)
\item[Prop.] 3: For any $\hat \permvar \in \permset_n $, it holds $\permset_n = \{\permvar \hat \permvar, \, \permvar\in \permset_n \}$; in other words,  for a given permutation $\hat \permvar$ playing the role of basis, $\permset_n$ can be expressed in terms of composed permutations.
\end{description}

A permutation $\permvar \in \permset_n$ is said to have a \emph{cycle} $\cyclevar \subseteq \batchset$ if the objects in $\cyclevar$ form an orbit under the sequence 
$l_{t+1} = \permvar(l_t)$, i.e.,  

\[
\permvar(l_t) = l_{t+1}  \quad  \text{ for } t = 1, \ldots, T-1, \quad \text{ and } \permvar(l_T) = l_1,
\]
where $l_t \in \cyclevar$ for all $t$ and $T$ is the orbit size (a natural number);
given a permutation $\permvar$, the partition of $\batchset$ into disjoint cycles is \emph{uniquely} determined apart from cyclic reordering of the elements within each cycle (see Figure~\ref{fig:cycle}).
Henceforth, we denote by $\permcyclecount(\permvar)$ the number of distinct cycles of $\permvar$.
In the example in equation \eqref{eq:experm}, there are two cycles, namely $(1, 3, 2)$, which corresponds to $\permvar(1) = 3$, $\permvar(3) = 2$, $\permvar(2) = 1$, and $(4)$, which corresponds to $\permvar(4) = 4$ (see Figure~\ref{fig:cycle}).

Suppose that all elements of $\permset_\batchsize$ are assigned probability $1/\batchsize!$, i.e., 
\[
\prob{\permvar}:= \prob{\text{One selects } \permvar} = \frac{1}{n!}, \qquad \text{ for all } \permvar \in \permset_n.
\]
Let $\permcyclecount_n$ denote the number of cycles of a random permutation with the above probability assignment.
It is shown in~\cite{Shepp:1966} that the number of cycles $\permcyclecount_\batchsize$ has expectation $\expectation{\permcyclecount_\batchsize} = \log(\batchsize) + O(1)$ and variance 
$\text{var}(\permcyclecount_\batchsize) = \log(\batchsize) + O(1)$; here $\log$ denotes the natural logarithm.

\begin{figure}
\centering
\includegraphics[width=.5\linewidth]{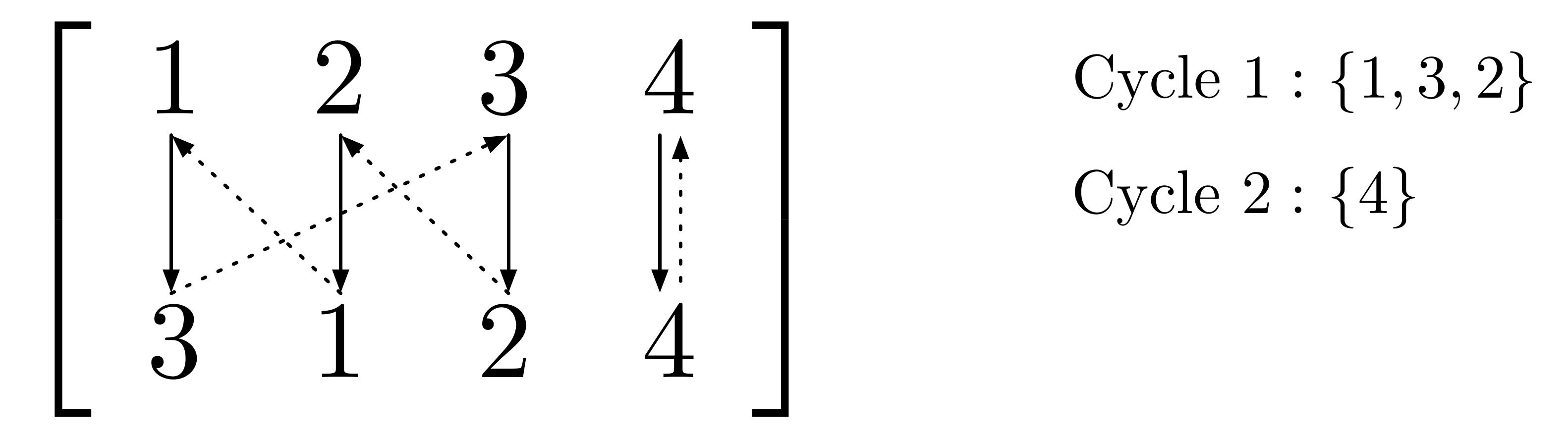}
\caption{The two cycles corresponding to the permutation: $\permvar(1) = 3$, $\permvar(2) = 1$, $\permvar(3) = 2$, and $\permvar(4) = 4$.
Cycle~1 can equivalently be expressed as $(2,1,3)$ or $(3,2,1)$.
Apart from this cyclic reordering, the decomposition into disjoint cycles is unique.}
\label{fig:cycle}
\end{figure}

\subsection{The Euclidean Bipartite Matching Problem}
\label{subsec:EBMP}

Let $\pickset_\numdem = \{x_1, \ldots, x_n\}$ and $\delvset_\numdem = \{y_1, \ldots, y_n\}$ be two sets of points in $\reals^\probdim$.
The Euclidean Bipartite Matching Problem is to find a permutation $\permvar^* \in \Pi_n$ (not necessarily unique)
such that the sum of the Euclidean distances between the matched pairs
$\{\text{$(y_i,x_{\permvar^*(i)})$ for $i = 1,\ldots,\numdem$}\}$
is minimized, i.e.:
\[
\sum_{i=1}^n \, \|x_{\permvar^*(i)} - y_i \| = \min_{\permvar \in \permset_n} \, \sum_{i=1}^n \, \|x_{\permvar(i)} - y_i \|,
\]
where $\|\cdot\|$ denotes the Euclidean norm and $\permset_\numdem$ denotes the set of all permutations over $\numdem$ elements.
Let $Q_\numdem := (\pickset_\numdem, \delvset_\numdem)$; we refer to the left-hand side in the above equation as the optimal bipartite matching cost $\M(Q_\numdem)$;
we refer to $\avglink_\text{M}(Q_\numdem) := \M(Q_\numdem)/\numdem$ as the \emph{average} match cost.

The Euclidean bipartite matching problem is generally solved by the ``Hungarian'' method~\cite{Kuhn1955}, which runs in $O(n^3)$ time.
The $O(n^3)$ barrier was indeed broken by Agarwal et al.~\cite{Agarwal:1995}, who presented a class of algorithms running in $O(\batchsize^{2+\epsilon})$, where $\epsilon$ is an arbitrarily small positive constant.
Additionally, there are also several approximation algorithms: among others, 
the algorithm presented in~\cite{Agarwal:2004} produces a $O(\log(1/\epsilon))$ optimal solution in expected runtime $O(\batchsize^{1+\epsilon})$, where, again, $\epsilon$ is an arbitrarily small positive constant.

The EBMP over \emph{random} sets of points enjoys some remarkable properties.
Specifically, let $\pickset_\numdem = \{ X_1, \ldots, X_n \}$ be a set of $n$ points in a compact set $\env \subset \reals^d$, $d\geq 3$,
that are independently, identically distributed (i.i.d.) according to a probability distribution with density $\den : \env \to \reals_{\geq 0}$;
let $\delvset_\numdem = \{ Y_1, \ldots, Y_n \}$ be a set of $n$ points in $\env$ that are i.i.d. according to the same probability distribution.
In \cite{VD-JEY:95} it is shown that there
exists a constant $\beta_{\mathrm{M},d}$ such that the optimal bipartite matching cost $\M(Q_\numdem) = \min_{\permvar \in \permset_n} \sum_{i=1}^n \| X_{\permvar(i)} - Y_i \|$ has limit behavior
\begin{equation}
\label{eq:EBMP_common}
\begin{aligned}
\lim_{n\rightarrow+\infty} \frac{\M(Q_\numdem)}{n^{1-1/d}} =
\beta_{\mathrm{M},d} \int_\env
\overline
\den(x)^{1-1/d}\;dx,
\end{aligned}
\end{equation}
almost surely,
where $\overline \den$ is the density of the absolutely continuous part of the point distribution.
The constant $\beta_{\mathrm{M},3}$ has been estimated numerically as $\beta_{\mathrm{M},3} \approx
0.7080 \pm 0.0002$~\cite{JH-JHBM-OCM:98}.

In the case $\probdim=2$ (i.e., the planar case) the following weaker result~\cite{TAL:92} holds
with high probability as $n \to +\infty$ (i.e. with probability $1-o(1)$):
\begin{equation}
\label{eq:matching2}
\M(Q_\batchsize) / (\batchsize\log\batchsize)^{1/2} \leq \gamma
\end{equation}
for some positive constant $\gamma$. If the probability distribution
is uniform, it also holds with high probability that $\M(Q_\batchsize) / (\batchsize\log\batchsize)^{1/2} $ is bounded below by a positive constant \cite{Ajtai:1984}.

To the best of our knowledge, there have been no similar results in the literature that apply when the distributions of $\mathcal X$ points is different from the distributions of $\mathcal Y$ points (which is the typical case for transportation systems).

\subsection{Euclidean Wasserstein distance}

As noted, little is known about the growth order of the EBMP matching when the distributions of $\mathcal X$ points is different from the distributions of $\mathcal Y$ points.
One of the main contributions of the paper is to extend the results in~\cite{VD-JEY:95} to this general case, for which the following notion of transportation complexity will prove useful.

Let $\den_1$ and $\den_2$ be two probability densities over $\env \subset \reals^\probdim$.
The \emph{Euclidean Wasserstein distance} between $\den_1$ and $\den_2$ is defined as
\begin{equation}
\label{eq:wasserstein}
  W(\den_1,\den_2)
  = \inf_{ \gamma \in \Gamma( \den_1,\den_2 ) } \int_{x,y \in \env} \vecnorm{y-x} \ d\gamma(x,y),
\end{equation}
where $\Gamma( \den_1,\den_2 )$ denotes the set of measures over the product space $\env \times \env$ having marginal densities $\den_1$ and $\den_2$, respectively.
The Euclidean Wasserstein distance is a continuous version of the so-called Earth Mover's distance; properties of the generalized version are discussed in~\cite{Ruschendorf1985}.

\subsection{The Strong Law of Absolute Differences}
\label{subsec:absolute_diff}
The last bit of background is a slightly more general version of the well-known Strong Law of Large Numbers (SLLN).
Let $X_1, \ldots, X_n$ be a sequence of scalar random variables that are i.i.d. with mean $\expect X$ and finite variance.
Then the sequence of cumulative sums $S_n = \sum_{i=1}^n X_i$ has the property (discussed, e.g., in~\cite{Baum:1965}) that
%
\begin{equation*}
\lim_{n\to\infty} \frac{ S_n - \expect S_n }{ n^\alpha } = 0, \quad \quad \text{almost surely},
\end{equation*}
for any $\alpha > 1/2$.
Note that the SLLN is the special case where $\alpha = 1$.

\section{Problem Statement} \label{sec:problem}
In this section, we first rigorously state the two routing problems that will be the subject of this paper, and then we state our objectives.

\subsection{The Euclidean Stacker Crane Problem}

Let $\pickset_\numdem = \{ x_1, \ldots, x_\numdem \}$ and
$\delvset_\numdem = \{ y_1, \ldots, y_\numdem \}$ be two sets of points in the $d$-dimensional Euclidean space $\reals^d$, where $d\geq 2$.
The Euclidean Stacker Crane Problem (ESCP) is to find a minimum-length tour through the points in $\pickset_\numdem \cup \delvset_\numdem$ with the property
that each point $x_i$ (which we call the $i$th \emph{pickup}) is immediately followed by the point $y_i$ (the $i$th \emph{delivery}); in other words, the pair $(x_i,y_i)$ must be visited in consecutive order (see Figure~\ref{fig:SCP}).
The length of a tour is the sum of all Euclidean distances along the tour.
We will refer to any feasible tour (i.e., one satisfying the pickup-to-delivery constraints) as a \emph{stacker crane tour}, and to the minimum-length such tour as the optimal stacker crane tour.
Note that the ESCP is a \emph{constrained} version of the well-known Euclidean Traveling Salesman Problem.

\begin{figure}
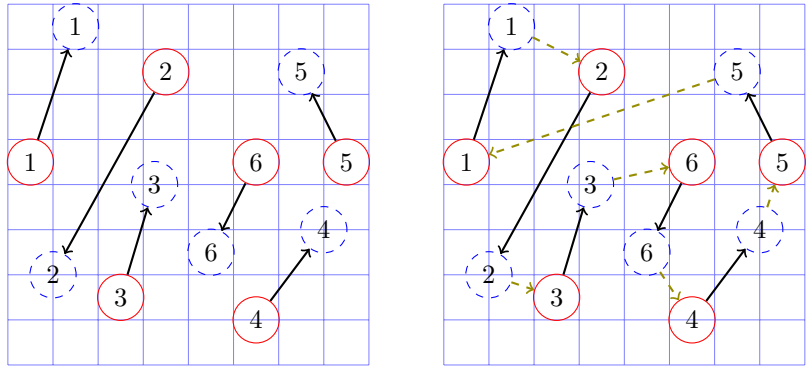

\centering
\subfigure[Six pickup/delivery pairs are generated in the Euclidean plane.]{
\begin{tourmaker}
\drawnodes
\drawlinks
\end{tourmaker}
}
\quad
\subfigure[Dashed arrows \emph{combined} with the solid arrows represent a stacker crane tour.]{
\begin{tourmaker}
\drawnodes
\drawlinks
\drawstackerplus
\end{tourmaker}
}
\caption{Example of Euclidean Stacker Crane Problem in two dimensions.
Solid and dashed circles denote pickup and delivery points, respectively; solid arrows denote the routes from pickup points to their delivery points.}
\label{fig:SCP}
\end{figure}

In this paper we focus on a \emph{stochastic} version of the ESCP.
Let $\pickset_\numdem = \{ X_1, \ldots, X_n \}$ be a set of points in a compact set $\env \subset \reals^d$
that are i.i.d. according to a distribution with density $\den_\picktag : \env \to \reals_{\geq 0}$;
let $\delvset_\numdem = \{ Y_1, \ldots, Y_n \}$ be a set of points in $\env$ that are i.i.d. according to a distribution with density $\den_\delvtag : \env \to \reals_{\geq 0}$.
To obtain the relevant transportation problem, we interpret each pair $(X_i,Y_i)$ as the pickup and delivery sites, respectively, of some transportation demand, and we seek to determine the cost of an optimal stacker crane tour through all points.
We will refer to this stochastic version of ESCP as $\ESCPinst$, and we will write $\pickset_\numdem,\delvset_\numdem \sim \ESCPinst$
to mean that $\pickset_\numdem$ contains $\numdem$ pickup sites i.i.d. with density $\den_\picktag$,
and $\delvset_\numdem$ contains $\numdem$ delivery sites i.i.d. with density $\den_\delvtag$.
An important contribution of this paper will be to characterize the behavior of the optimal stacker crane tour through $\pickset_\numdem$ and $\delvset_\numdem$ as a function of the parameters $n$, $\den_\picktag$, and $\den_\delvtag$;
throughout the paper we will assume that densities $\pdistr_\picktag$ and $\pdistr_\delvtag$ are absolutely continuous.

Despite a close relation between the stochastic ESCP and the stochastic EBMP, the asymptotic cost of a stochastic ESCP has not been characterized to date.

\subsection{Dynamic Pickup Delivery Problems with Unit-Capacity Vehicles}
The 1-DPDP is defined as follows. A total of $m$ vehicles travel at unit velocity within a workspace $\env$; the vehicles have unlimited range and \emph{unit} capacity (i.e., they can transport at most one demand at a time).
Demands are generated according to a time-invariant Poisson process, with time intensity $\lambda \in \reals_{>0}$.
A newly arrived demand has an associated pickup location which is independent and identically distributed in $\env$ according to a density $\den_\picktag$.
Each demand must be transported from its pickup location to its delivery location, at which time it is removed from the system.
The delivery locations are also i.i.d. in $\env$ according to a density $\den_\delvtag$. 
A policy for routing the vehicles is said to be \emph{stabilizing} if the expected number of demands in the system remains uniformly bounded at all times; the objective is to find a stabilizing and causal routing policy that minimizes the asymptotic expected waiting times (i.e., the elapsed time between the arrival of a demand and its delivery) of demands.

This problem has been studied in~\cite{Treleaven.Pavone.ea:10} under the restrictive assumptions $\pdistr_\delvtag = \pdistr_\picktag := \varphi$ and $d\geq 3$;
in that paper, it has been shown that if one defines the ``load factor'' as
\[
\utilization \doteq \lambda\,  \expect_{\varphi} \vecnorm{Y-X}/m,
\]
where $Y$ and $X$ are two independent random points in $\env$ with a distribution of density $\den$, then the condition $\utilization < 1$ is necessary and sufficient for a stabilizing policy to exist. However, that analysis---and indeed the result itself---is no longer valid if  $\pdistr_\delvtag \neq \pdistr_\picktag$. This paper will show how the definition of load factor has to be modified for the more realistic case  $\pdistr_\delvtag \neq \pdistr_\picktag$. Pivotal in our approach is to characterize, with almost sure analytical bounds, the scaling of the optimal solution of $\ESCPinst$ with respect to the problem size.

\subsection{Objectives of the Paper}
In this paper we aim at solving the following three problems: 
\begin{description}
\item[\bf P1] Find a polynomial-time algorithm $\mathcal A$ for the ESCP which is asymptotically optimal almost surely, i.e.,
\[\lim_{n\to+\infty} L_{\text{$\mathcal A$}}(n) / L^*(n)
= 1,\]
where $n$ is the size (number of demands) of the stochastic instance, $L_{\text{$\mathcal A$}}(n)$ is the length of the stacker crane tour produced by algorithm $\mathcal A$, and $L^*(n)$
is the length of the optimal stacker crane tour.
\item[\bf P2] For the general case $\pdistr_\delvtag \neq \pdistr_\picktag$, characterize the growth (with respect to the problem size) in the cost of the ESCP with almost sure analytical bounds. 
\item[\bf P3] Find a necessary and sufficient condition for the existence of stabilizing policies
for the 1-DPDP as a function of the problem parameters,
i.e., $\arrivalrate$, $\numveh$, $\den_\picktag$, $\den_\delvtag$, $\env$.
\end{description}  

The solutions to the above three problems collectively lead to a new class of robust, polynomial-time, provably-efficient algorithms for vehicle routing in large-scale transportation systems.

\section{An Asymptotically Optimal Polynomial-Time Algorithm for the Stochastic ESCP}
\label{sec:scpoptimal}

In this section we present an asymptotically optimal, polynomial-time algorithm for the stochastic ESCP, which we call~\ref{alg:splicing}.
The key idea behind \ref{alg:splicing} is to connect the tour from delivery sites back to pickup sites in accordance with an optimal bipartite matching between the sets of pickup and delivery sites.
Unfortunately, this procedure is likely to generate a certain number of \emph{disconnected subtours} (see Figure~\ref{subfig:sub_disc}), and so, in general, the result is not a stacker crane tour.
The key property we prove is that the number of disconnected subtours
grows quite slowly, i.e., sublinearly, with the size of the problem.
Then, by using a greedy algorithm to connect such subtours, one obtains an asymptotically optimal solution to the ESCP with a polynomial number of operations
(since an optimal matching can be computed in polynomial time).

\subsection{The \ref{alg:splicing} Algorithm}

The algorithm~\ref{alg:splicing} is described in pseudo-code.
In line \ref{line:perm}, the algorithm $\algog$ is \emph{any} algorithm that computes optimal bipartite matchings. After the pickup-to-delivery links and the optimal bipartite matching links are added (lines \ref{line:pd}-\ref{line:ml}), there might be a number of disconnected subtours (they do, however, satisfy the pickup-to-delivery contraints).
In that case (i.e., when $\subn>1$), links between subtours are added, e.g., by using a nearest-neighbor rule\footnote{
In this paper we use a simple nearest-neighbor heuristic (in~\ref{alg:splicing}, lines~\ref{line:connectingheuristic} to~\ref{line:endconnectingheuristic}) for adding connecting links to form an SCP tour.
However, the results of this paper do not depend on this choice, and \emph{any} connecting heuristic can be used. 
}.
\ifthenelse{ \isundefined{\hidesplicesample} }{
Figure~\ref{fig:matchsplice} shows a sample execution of the algorithm; we refer to the delivery-to-pickup links added in lines \ref{line:con} and \ref{line:clo} (the green links in Figure~\ref{fig:matchsplice}) as \emph{connecting links}, since they connect  the subtours.
}{
We refer to the delivery-to-pickup links added in lines \ref{line:con} and \ref{line:clo} as \emph{connecting links}, since they connect the subtours.
}
%
The complexity of~\ref{alg:splicing} is dominated by the construction of the optimal bipartite matching, which takes time $O(\batchsize^{2+\epsilon})$.
\begin{algorithm}
\renewcommand{\thealgorithm}{SPLICE}
\caption{}	
\label{alg:splicing}
\begin{algorithmic}[1]
\renewcommand{\algorithmicrequire}{\textbf{Input:}}
\renewcommand{\algorithmicensure}{\textbf{Output:}}
\renewcommand{\algorithmiccomment}[1]{ // \emph{#1} }
\REQUIRE a set of \emph{demands} $\batchset = \left\{ (\xvar_1,\yvar_1),\ldots,(\xvar_\batchsize,\yvar_\batchsize) \right\}$, $n>1$.
\ENSURE a Stacker Crane tour through $\batchset$.
\STATE {\bf initialize} $\permvar \gets$ solution to Euclidean bipartite matching problem between sets
  $\mathcal X=\left\{\xvar_1,\ldots,\xvar_\batchsize\right\}$ and
  $\mathcal Y=\left\{\yvar_1,\ldots,\yvar_\batchsize\right\}$ computed by using a bipartite matching algorithm $\algog$. \label{line:perm}
\STATE Add the $n$ pickup-to-delivery links $x_i \to y_i$, $i=1, \ldots, n$. \label{line:pd}
\STATE Add the $n$ matching links $y_i \to x_{\permvar(i)}$, $i=1, \ldots, n$. \label{line:ml}
\STATE $\subn \gets$ number of disconnected subtours.  \label{line:disc}
\IF{$\subn>1$} \label{line:connectingheuristic}
\STATE Arbitrarily order the $\subn$ subtours $\mathcal S_j$, $j=1,\ldots,n$, into an ordered set $\mathcal S:= \{\mathcal S_1, \ldots, \mathcal S_\subn\}$.
\STATE \texttt{base} $\gets $ index of an arbitrary delivery site in $\mathcal S_1$.
\STATE \texttt{prev} $\gets$ \texttt{base}.
  \FOR{$k=1 \to \subn-1$} \label{line:for}
 \STATE  Remove link $y_{\text{\texttt{prev}}}\to x_{\text{\permvar(\texttt{prev}})}$.\label{line:beg}  
 \STATE \texttt{next} $\gets$ index of pickup site in $\mathcal S_{k+1}$ that is closest to $y_{\text{\texttt{prev}}}$.
 \STATE Add link $y_{\text{\texttt{prev}}} \to x_{\text{\texttt{next}}}$. \label{line:con}
 \STATE \texttt{prev} $\gets \permvar^{-1}(\texttt{next})$. 
  \ENDFOR \label{line:efor}
\STATE Add link $y_{\text{\texttt{prev}}} \to x_{\permvar(\text{\texttt{base}})}$. \label{line:clo}
 \ENDIF  \label{line:endconnectingheuristic}
 \end{algorithmic}
\end{algorithm}

\ifthenelse{ \isundefined{\hidesplicesample} }{

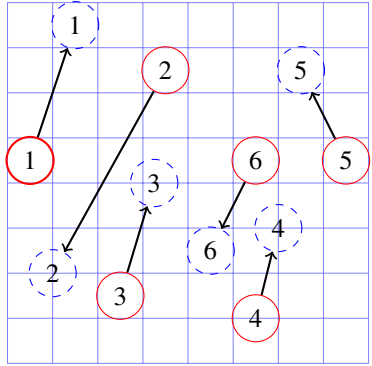
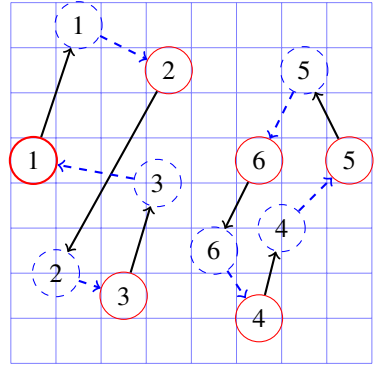
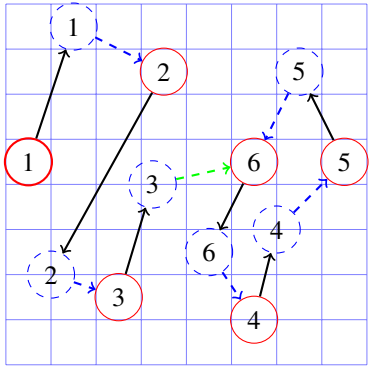
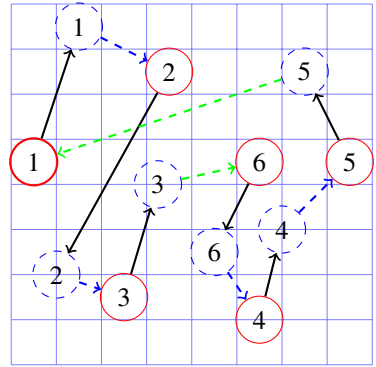
\begin{figure}
\centering
\subfigure[Line \ref{line:pd}: 6 pickup-to-delivery links are added.]{
\definecolor{gold}{rgb}{.85,.64,.125}

\begin{tikzpicture}[scale=.6]
\tikzstyle{pickup}=[circle,draw=red!100] 
\tikzstyle{delivr}=[circle,dashed,draw=blue!100] 
\tikzstyle{p2d}=[->,thick]
\tikzstyle{stacker}=[->,thick,dashed,color=gold]
\tikzstyle{match}=[->,thick,dashed,color=blue]
\tikzstyle{primary}=[color=green,->,thick,dashed]
\tikzstyle{node-done} = [circle,draw=gray!100] 
\tikzstyle{edge-done} = [color=gray,->,thick]

\draw[step=1,color=blue!50,very thin] (-4,-4) grid (4,4) ;		

\coordinate (p1) at (-3.5,.5) ;		\coordinate (d1) at (-2.5, 3.5) ;	
\coordinate (p2) at (-.5,2.5) ;		\coordinate (d2) at (-3, -2) ;		
\coordinate (p3) at (-1.5,-2.5) ;	\coordinate (d3) at (-.75, 0) ;		
\coordinate (p4) at (1.5,-3) ;		\coordinate (d4) at (2, -1) ;		
\coordinate (p5) at (3.5,.5) ;		\coordinate (d5) at (2.5, 2.5) ;	
\coordinate (p6) at (1.5,.5) ;		\coordinate (d6) at (.5, -1.5) ;


\node [delivr] (delv1) at (d1) {1} ;
\node [pickup,thick] (pick1) at (p1) {1} ; \draw [p2d] (pick1) -- (delv1) ;
\node [delivr] (delv2) at (d2) {2} ;
\node [pickup] (pick2) at (p2) {2} ; \draw [p2d] (pick2) -- (delv2) ;
\node [delivr] (delv3) at (d3) {3} ;
\node [pickup] (pick3) at (p3) {3} ; \draw [p2d] (pick3) -- (delv3) ;
\node [delivr] (delv4) at (d4) {4} ;
\node [pickup] (pick4) at (p4) {4} ; \draw [p2d] (pick4) -- (delv4) ;
\node [delivr] (delv5) at (d5) {5} ;
\node [pickup] (pick5) at (p5) {5} ; \draw [p2d] (pick5) -- (delv5) ;
\node [delivr] (delv6) at (d6) {6} ;
\node [pickup] (pick6) at (p6) {6} ; \draw [p2d] (pick6) -- (delv6) ;




\end{tikzpicture}
}\qquad \qquad \qquad
\subfigure[Line \ref{line:ml}: 6 matching links are added. The number of disconnected subtours is $N=2$.]{
\definecolor{gold}{rgb}{.85,.64,.125}

\begin{tikzpicture}[scale=.6]
\tikzstyle{pickup}=[circle,draw=red!100] 
\tikzstyle{delivr}=[circle,dashed,draw=blue!100] 
\tikzstyle{p2d}=[->,thick]
\tikzstyle{stacker}=[->,thick,dashed,color=gold]
\tikzstyle{match}=[->,thick,dashed,color=blue]
\tikzstyle{primary}=[color=green,->,thick,dashed]
\tikzstyle{node-done} = [circle,draw=gray!100] 
\tikzstyle{edge-done} = [color=gray,->,thick]

\draw[step=1,color=blue!50,very thin] (-4,-4) grid (4,4) ;		

\coordinate (p1) at (-3.5,.5) ;		\coordinate (d1) at (-2.5, 3.5) ;	
\coordinate (p2) at (-.5,2.5) ;		\coordinate (d2) at (-3, -2) ;		
\coordinate (p3) at (-1.5,-2.5) ;	\coordinate (d3) at (-.75, 0) ;		
\coordinate (p4) at (1.5,-3) ;		\coordinate (d4) at (2, -1) ;		
\coordinate (p5) at (3.5,.5) ;		\coordinate (d5) at (2.5, 2.5) ;	
\coordinate (p6) at (1.5,.5) ;		\coordinate (d6) at (.5, -1.5) ;


\node [delivr] (delv1) at (d1) {1} ;
\node [pickup,thick] (pick1) at (p1) {1} ; \draw [p2d] (pick1) -- (delv1) ;
\node [delivr] (delv2) at (d2) {2} ;
\node [pickup] (pick2) at (p2) {2} ; \draw [p2d] (pick2) -- (delv2) ;
\node [delivr] (delv3) at (d3) {3} ;
\node [pickup] (pick3) at (p3) {3} ; \draw [p2d] (pick3) -- (delv3) ;
\node [delivr] (delv4) at (d4) {4} ;
\node [pickup] (pick4) at (p4) {4} ; \draw [p2d] (pick4) -- (delv4) ;
\node [delivr] (delv5) at (d5) {5} ;
\node [pickup] (pick5) at (p5) {5} ; \draw [p2d] (pick5) -- (delv5) ;
\node [delivr] (delv6) at (d6) {6} ;
\node [pickup] (pick6) at (p6) {6} ; \draw [p2d] (pick6) -- (delv6) ;

\draw [match] (delv1) -- (pick2) ;
\draw [match] (delv2) -- (pick3) ;
\draw [match] (delv3) -- (pick1) ;
\draw [match] (delv4) -- (pick5) ;
\draw [match] (delv5) -- (pick6) ;
\draw [match] (delv6) -- (pick4) ;



\end{tikzpicture}
\label{subfig:sub_disc}}\\
\subfigure[Line \ref{line:beg}. Algorithm state: \texttt{prev} = \texttt{base} = 3, $k=1$. The link $y_3 \to x_1$ is removed, \texttt{next} is assigned the value 6, the link $y_3 \to x_6$ is added, \texttt{prev} is assigned the value 5.]{
\definecolor{gold}{rgb}{.85,.64,.125}

\begin{tikzpicture}[scale=.6]
\tikzstyle{pickup}=[circle,draw=red!100] 
\tikzstyle{delivr}=[circle,dashed,draw=blue!100] 
\tikzstyle{p2d}=[->,thick]
\tikzstyle{stacker}=[->,thick,dashed,color=gold]
\tikzstyle{match}=[->,thick,dashed,color=blue]
\tikzstyle{primary}=[color=green,->,thick,dashed]
\tikzstyle{node-done} = [circle,draw=gray!100] 
\tikzstyle{edge-done} = [color=gray,->,thick]

\draw[step=1,color=blue!50,very thin] (-4,-4) grid (4,4) ;		

\coordinate (p1) at (-3.5,.5) ;		\coordinate (d1) at (-2.5, 3.5) ;	
\coordinate (p2) at (-.5,2.5) ;		\coordinate (d2) at (-3, -2) ;		
\coordinate (p3) at (-1.5,-2.5) ;	\coordinate (d3) at (-.75, 0) ;		
\coordinate (p4) at (1.5,-3) ;		\coordinate (d4) at (2, -1) ;		
\coordinate (p5) at (3.5,.5) ;		\coordinate (d5) at (2.5, 2.5) ;	
\coordinate (p6) at (1.5,.5) ;		\coordinate (d6) at (.5, -1.5) ;


\node [delivr] (delv1) at (d1) {1} ;
\node [pickup,thick] (pick1) at (p1) {1} ; \draw [p2d] (pick1) -- (delv1) ;
\node [delivr] (delv2) at (d2) {2} ;
\node [pickup] (pick2) at (p2) {2} ; \draw [p2d] (pick2) -- (delv2) ;
\node [delivr] (delv3) at (d3) {3} ;
\node [pickup] (pick3) at (p3) {3} ; \draw [p2d] (pick3) -- (delv3) ;
\node [delivr] (delv4) at (d4) {4} ;
\node [pickup] (pick4) at (p4) {4} ; \draw [p2d] (pick4) -- (delv4) ;
\node [delivr] (delv5) at (d5) {5} ;
\node [pickup] (pick5) at (p5) {5} ; \draw [p2d] (pick5) -- (delv5) ;
\node [delivr] (delv6) at (d6) {6} ;
\node [pickup] (pick6) at (p6) {6} ; \draw [p2d] (pick6) -- (delv6) ;

\draw [match] (delv1) -- (pick2) ;
\draw [match] (delv2) -- (pick3) ;
\draw [match] (delv4) -- (pick5) ;
\draw [match] (delv5) -- (pick6) ;
\draw [match] (delv6) -- (pick4) ;
\draw [primary] (delv3) -- (pick6) ;



\end{tikzpicture}
}\qquad \qquad \qquad
\subfigure[Line \ref{line:clo}. Algorithm state: \texttt{prev} = 5, \texttt{base} = 3. The link $y_5\to x_1$ is added an the tour is completed.]{
\definecolor{gold}{rgb}{.85,.64,.125}

\begin{tikzpicture}[scale=.6]
\tikzstyle{pickup}=[circle,draw=red!100] 
\tikzstyle{delivr}=[circle,dashed,draw=blue!100] 
\tikzstyle{p2d}=[->,thick]
\tikzstyle{stacker}=[->,thick,dashed,color=gold]
\tikzstyle{match}=[->,thick,dashed,color=blue]
\tikzstyle{primary}=[color=green,->,thick,dashed]
\tikzstyle{node-done} = [circle,draw=gray!100] 
\tikzstyle{edge-done} = [color=gray,->,thick]

\draw[step=1,color=blue!50,very thin] (-4,-4) grid (4,4) ;		

\coordinate (p1) at (-3.5,.5) ;		\coordinate (d1) at (-2.5, 3.5) ;	
\coordinate (p2) at (-.5,2.5) ;		\coordinate (d2) at (-3, -2) ;		
\coordinate (p3) at (-1.5,-2.5) ;	\coordinate (d3) at (-.75, 0) ;		
\coordinate (p4) at (1.5,-3) ;		\coordinate (d4) at (2, -1) ;		
\coordinate (p5) at (3.5,.5) ;		\coordinate (d5) at (2.5, 2.5) ;	
\coordinate (p6) at (1.5,.5) ;		\coordinate (d6) at (.5, -1.5) ;


\node [delivr] (delv1) at (d1) {1} ;
\node [pickup,thick] (pick1) at (p1) {1} ; \draw [p2d] (pick1) -- (delv1) ;
\node [delivr] (delv2) at (d2) {2} ;
\node [pickup] (pick2) at (p2) {2} ; \draw [p2d] (pick2) -- (delv2) ;
\node [delivr] (delv3) at (d3) {3} ;
\node [pickup] (pick3) at (p3) {3} ; \draw [p2d] (pick3) -- (delv3) ;
\node [delivr] (delv4) at (d4) {4} ;
\node [pickup] (pick4) at (p4) {4} ; \draw [p2d] (pick4) -- (delv4) ;
\node [delivr] (delv5) at (d5) {5} ;
\node [pickup] (pick5) at (p5) {5} ; \draw [p2d] (pick5) -- (delv5) ;
\node [delivr] (delv6) at (d6) {6} ;
\node [pickup] (pick6) at (p6) {6} ; \draw [p2d] (pick6) -- (delv6) ;

\draw [match] (delv1) -- (pick2) ;
\draw [match] (delv2) -- (pick3) ;
\draw [match] (delv4) -- (pick5) ;
\draw [match] (delv6) -- (pick4) ;
\draw [primary] (delv3) -- (pick6) ;
\draw [primary] (delv5) -- (pick1) ;



\end{tikzpicture}
}
\caption{Sample execution of the~\ref{alg:splicing} algorithm. The solution to the EBMP is $\permvar(1) = 2$, $\permvar(2) = 3$, $\permvar(3) = 1$, $\permvar(4) = 5$, $\permvar(5) = 6$, and $\permvar(6) = 4$.
Demands are labeled with integers.
Pickup and delivery sites are represented by solid and dashed circles, respectively.
Pickup-to-delivery links are shown as black arrows.
Matching links are dark dashed arrows.
Subtour connections are shown as lighter, dashed arrows.
The resulting tour is $1\to2\to3\to6\to4\to5\to1$.}
\label{fig:matchsplice}
\end{figure}

}{}	

\subsection{Asymptotic Optimality of \ref{alg:splicing}}
\label{subsec:asymptotic_optimality}

In general the \ref{alg:splicing} algorithm produces a number of \emph{connecting links} between disconnected subtours (i.e., in general $\subn>1$; see Figure~\ref{fig:matchsplice}).
Thus a first step in proving asymptotic optimality of the \ref{alg:splicing}~algorithm is to characterize the growth order for the number of subtours with respect to $n$, the size of the problem instance, and thus the cost of the extra connecting links.
First we observe an equivalence between the number of subtours $\subn$ produced by line~\ref{line:ml} and the number of cycles for the permutation $\permvar$ in line \ref{line:perm}.
\begin{lemma}[Permutation cycles and subtours]\label{lemma:sub2perm}
The number $\subn$ of subtours produced by the \ref{alg:splicing} algorithm in line \ref{line:ml} is equal to $\permcyclecount(\permvar)$,
where $\permcyclecount(\permvar)$ is the number of cycles of the permutation $\permvar$ computed in line \ref{line:perm}.
\end{lemma}
\begin{proof}
Let $\mathcal D_k$ be the set of delivery sites for subtour $k$ ($k=1,\ldots,\subn$).
By construction, the indices in $\mathcal D_k$ constitute a cycle of the permutation $\permvar$. For example, in Figure~\ref{fig:matchsplice}, the indices of the delivery sites in the subtour $x_1\to y_1\to x_2 \to y_2 \to x_3 \to y_3 \to x_1$ are $\{1,2,3\}$, and they constitute a cycle for $\permvar$ since $\permvar(1)=2$, $\permvar(2)=3$, and $\permvar(3)=1$.
Since the subtours are disconnected, and every index is contained by some subtour, then the sets $\mathcal D_k$ ($k=1,\ldots,\subn$) represent a partition of $\{1,\ldots,n\}$ into the disjoint cycles of $\permvar$.
This implies that the number of subtours $\subn$ is equal to $\permcyclecount(\permvar)$.  
\end{proof}
By the lemma above, characterizing the number of subtours generated during the execution of the algorithm is equivalent to characterizing the number of cycles for the permutation $\permvar$.
Leveraging  the i.i.d. structure in our problem setup, one can argue intuitively that all permutations should be equiprobable.
In fact, the statement withstands rigorous proof.
\begin{lemma}[Equiprobability of permutations]\label{lemma:equi}
Let $Q_\numdem = (\pickset_\numdem, \delvset_\numdem)$ be a random instance of the EBMP, where $\pickset_\numdem,\delvset_\numdem \sim \ESCPinst$.
Then
\[
\prob{\sigma} =  \frac{1}{n!}	\qquad \text{for all $\sigma \in \permset$,}
\]
where $\prob{\sigma}$ denotes the probability that an optimal bipartite matching algorithm $\algog$ produces as a result the permutation $\sigma$.
\end{lemma}
\ifSHORTversion
\begin{proof}
In this paper we provide a summary of the proof. For the complete proof we refer to the extended copy of the paper on arXiv.

Let $\{x_1, \ldots, x_n\}$ and $\{y_1, \ldots, y_n\}$ be sample realizations of $\pickset_n$ and $\delvset_n$, respectively.
Let $\batchvar = \concat( \xvar_1, \yvar_1, \ldots, \xvar_\batchsize, \yvar_\batchsize )$
denote the column vector formed by vertical concatenation of $x_1, y_1, \ldots, x_n, y_n$.
The set $\batchsetdesc$ of possible vectors, i.e. the span of the instances of the EBMP, is a full-dimensional subset of $\reals^{d(2n)}$.
For each permutation $\permvar \in \permset_n$, let us define the set
\[
  \batchset_{\permvar} := \left\{  \batchvar \ \in \batchsetdesc \,|\, \text{$\permvar$ is the unique optimal matching of $\batchvar$} \right\};
\]
let us denote by $\zb$ the (remaining) set of batches that lead to non-unique optimal bipartite matchings.
The sets $\{ \batchset_{\permvar} \}_{\permvar\in\permset_\batchsize}$, along with $\zb$, form a partition of $\batchsetdesc$.
It can be shown that the set $\zb$ has probability zero,
so that for any permutation $\hat \permvar$ we have
\begin{equation}\label{eq:permutation_probability_expression}
\prob{\hat \permvar}
  = \prob{\hat s \in \batchset_{\hat \permvar}}
  = \int_{\hat s \in \batchset_{\hat \permvar}} \pdistr(\hat s) d{\hat s},
\end{equation}
where $\pdistr(\hat s)$ denotes 
$\prod_{i=1}^n \pdistr_\picktag(\hat x_i) \pdistr_\delvtag(\hat y_i)$.
We can write a \emph{reordering function} $g_{\hat\permvar}: \batchsetdesc \to \batchsetdesc$ as the invertible function that maps a batch
$\batchvar = \concat(x_1,y_1,\ldots,x_n, y_n)$
into a batch $\batchvar' = \concat(x_{\hat\permvar(1)}, y_1,\ldots,x_{\hat\permvar(n)}, y_n)$.
It can be shown that
$g_{\hat \permvar}(\batchset_{\hat \permvar}) = \batchset_{\permvar_1}$,
where
$g_{\hat \permvar}(\batchset_{\hat \permvar}) := \{ g_{\hat\permvar}(\hat\batchvar) \ : \ \hat\batchvar \in \batchset_{\hat \permvar} \}$.
(Recall that $\permvar_1$ denotes the \emph{identity} permutation.)
It can also be shown that $\pdistr(\cdot) = \pdistr(g_{\hat\permvar}(\cdot))$.
Using these two properties and a variable substitution $s = g_{\hat \permvar}(\hat s)$ in~\eqref{eq:permutation_probability_expression},
it can be shown that
\begin{equation*}
\begin{split}
\prob{ \hat s \in \batchset_{\hat \permvar} } =
\prob{ s \in \batchset_{\permvar_1} } = \prob{\permvar_1}.
\end{split}
\end{equation*}
We obtain the lemma by repeating the argument for all $\hat \permvar \in \permset_n$.
\end{proof}

\else
\begin{proof}
See Appendix~\ref{appendix:numberofcycles} for the proof.
\end{proof}
\fi

Lemmas~\ref{lemma:sub2perm} and~\ref{lemma:equi} allow us to apply the result in Section~\ref{subsec:cycle} to characterize the growth order for the number of subtours;
in particular, we observe that
$\subn = \permcyclecount(\permvar) = \log(n) + O(1)$.
For the proof of optimality for \ref{alg:splicing} we would like to make a slightly stronger statement:
\begin{lemma}[Asymptotic number of subtours]
\label{lemma:splicesvanish}
Let $d\geq 2$, and let
$Q_\numdem = (\pickset_\numdem, \delvset_\numdem)$ be a random instance of the EBMP, where $\pickset_\numdem,\delvset_\numdem \sim \ESCPinst$.
Let $\subn_n$ be the number of subtours generated by the \ref{alg:splicing} algorithm on the problem instance $Q_\numdem$.
Then
\[
\lim_{n\to +\infty} \subn_n/n = 0,
\]
almost surely.
\end{lemma}
\ifSHORTversion
\begin{proof}[Proof of Lemma~\ref{lemma:splicesvanish}]
By Lemma~\ref{lemma:sub2perm}, the number of disconnected subtours is \emph{equal} to the number of cycles in the permutation $\permvar$, computed by the matching algorithm $\algog$ in line \ref{line:perm}.
Since, by Lemma \ref{lemma:equi}, all permutations are equiprobable, the number of cycles has expectation and variance both equal to $\log(n) + O(1)$.
The remainder of the proof is a straightforward application of the Borel-Cantelli lemma; for a detailed proof, see the arXiv version of the paper.
\end{proof}

\else
\begin{proof}
See Appendix~\ref{appendix:lemmas} for the proof.
\end{proof}
\fi
\begin{remark}\label{rem:van}
For $d\geq 3$, one can similarly prove that $\lim_{n\to +\infty} \subn_n/n^{1-1/d}=0$ almost surely. 
\end{remark}

Having characterized the number of subtours generated by~\ref{alg:splicing}, we are now ready to prove the main result of the section, i.e., the asymptotic optimality of the algorithm.
\begin{theorem} \label{thm:asympoptimality}
Let $d\geq 2$. Let $\mathcal X_n$ be a set of points $\{X_1, \ldots,X_n\}$ that are i.i.d. in a compact set $\env \subset \reals^d$ and distributed according to a density $\den_{\text{P}}$; let $\mathcal Y_n$ be a set of points $\{Y_1, \ldots,Y_n\}$ that are i.i.d. in a compact set $\env \subset \reals^d$ and distributed according to a density $\den_{\text{D}}$. 
Then
\[
\lim_{\batchsize\to +\infty}
  \frac{ \tourlen_\text{SPLICE}(\batchsize)
  }{
    \tourlen^*(\batchsize)
  } = 1, \quad \text{almost surely}.
\]
\end{theorem}
\bigskip
\begin{proof}
Let $Q_n = (\mathcal X_n, \mathcal Y_n)$.
The length of the optimal stacker crane tour through the pickup points $\mathcal X_n$ and the delivery points $\mathcal Y_n$ is bounded below by
\begin{equation}\label{eq:SClower}
  \tourlen^*(\batchsize)
  \geq \sum_{i=1}^\numdem \vecnorm{\Yvar_i - \Xvar_i} + \M(Q_n).
\end{equation}
On the other hand, the number of connecting links added by the \ref{alg:splicing} algorithm
is bounded above by the number of subtours $\subn_n$ of the optimal bipartite matching, and the length of any connecting link is bounded above by $\max_{x,y\in\env}\| x-y\|$.
Hence, $\tourlen_{\text{SPLICE}}(\batchsize)$ can be bounded above by
\begin{align*}
  \tourlen_\text{SPLICE}(\batchsize)
    &\leq \sum_{i=1}^\numdem \vecnorm{\Yvar_i - \Xvar_i} + \M(Q_n) + \max_{x,y\in\env}\| x-y\| \, \subn_n 		\\
    &\leq \tourlen^*(\batchsize)
    + \max_{x,y\in\env}\| x-y\|\, \subn_n.
\end{align*}
By the Strong Law of Large Numbers, $\lim_{n\to +\infty} \sum_i\vecnorm{\Yvar_i - \Xvar_i}/n = \expect_\spatialjoint {\vecnorm{\Yvar - \Xvar}} $ almost surely.
Hence, $\tourlen^*(\batchsize)$ has linear growth.
Since $\lim_{n\to +\infty}\subn_n/\batchsize = 0$ (by Lemma~\ref{lemma:splicesvanish}), one obtains the claim.
\end{proof}

\section{Analytical Bounds on the Cost of the ESCP} \label{sec:scpbounds}

In this section we derive analytical bounds on the cost of the optimal stacker crane tour.
The resulting bounds are useful for two reasons: (i) they give further insight into the ESCP (and the EBMP), and (ii) they will allow us to find a necessary and sufficient stability condition for our model of DRT systems  (i.e., for the 1-DPDP).

The development of these bounds follows from an analysis of the growth order, with respect to the instance size $\numdem$, of the EBMP matching on $Q_\numdem = (\pickset_\numdem, \delvset_\numdem)$, where $\pickset_\numdem,\delvset_\numdem \sim \ESCPinst$.
The main technical challenge is to extend the results in~\cite{VD-JEY:95}, about the length of the matching to the case where $\den_{\text{P}}$ and $\den_{\text{D}}$ are not identical.
We first derive in Section~\ref{subsec:linklower} a lower bound on the length of the EBMP matching for the case $\den_{\text{P}}\neq \den_{\text{D}}$ (and resulting lower bound for the ESCP);
then in Section~\ref{subsec:linkupper} we find the corresponding upper bounds.

\subsection{A Lower Bound on the Length of the ESCP}
\label{subsec:linklower}

In the rest of the paper, we let $\cellset = \{ \workcell^1, \ldots, \workcell^{|\cellset|} \}$ denote some finite partition of Euclidean environment $\env$ into $|\cellset|$ \emph{cells}.
We denote by
\begin{equation*}
	\pdistr_\picktag(\workcell^\ivar)
	:= \int_{ \locvar \in \workcell^\ivar } \pdistr_\picktag(\locvar) d\locvar
\end{equation*}
the measure of cell $\workcell^\ivar$ under the pickup distribution (with density $\pdistr_\picktag$), i.e., the probability that a particular pickup $X$ is in the $\ivar$th cell.
Similarly, we denote by
\begin{equation*}
	\pdistr_\delvtag(\workcell^\ivar)
	:= \int_{ \locvartwo \in \workcell^\ivar } \pdistr_\delvtag(\locvartwo) d\locvartwo
\end{equation*}
the cell's measure under the delivery distribution (with density $\pdistr_\delvtag$), i.e., the probability that a particular delivery $Y$ is in the $\ivar$th cell.
Most of the results of the paper are valid for arbitrary partitions of the environment; however, for some of the more delicate analysis we will refer to the following particular construction.
Without loss of generality, we assume that the environment $\env \subset \reals^d$ is a hyper-cube with side-length $\sidelength$.
For some integer $\sidecuts \geq 1$, we construct a partition $\cellset_\sidecuts$ of $\env$ by slicing the hyper-cube into a grid of $\sidecuts^\probdim$ smaller cubes, each length $\sidelength/\sidecuts$ on a side; inclusion of subscript $\sidecuts$ in our notation will make the construction explicit.
The ordering of cells in $\cellset_\sidecuts$ is arbitrary.

Our first result bounds the average match length $l_\text{M}(Q_\numdem)$ asymptotically from below.
In preparation for this result we present Problem~\ref{problem:traffic_optimistic}, a linear optimization problem whose solution maps partitions to real numbers.
\begin{problem}[Optimistic ``rebalancing'']
\label{problem:traffic_optimistic}
\begin{equation*}	
\begin{aligned}
	& \underset{
	  \left\{ \alpha_{\ivar\jvar} \geq 0 \right\}_{ i,j \in \{1,\ldots,r^d\} } }{\text{Minimize}}
	& & \sum_{\ivar\jvar} \alpha_{\ivar\jvar} \min_{y\in\workcell^\ivar, x\in\workcell^\jvar} \vecnorm{x-y} \\
	& \text{subject to}
	& &		\sum_\jvar \alpha_{\ivar\jvar} = \pdistr_\delvtag( \workcell^\ivar )
			\quad \text{for all $\workcell^\ivar \in \cellset$,}
	\\
	&&&		\sum_\ivar \alpha_{\ivar\jvar} = \pdistr_\picktag( \workcell^\jvar )
			\quad \text{for all $\workcell^\jvar \in \cellset$.}
\end{aligned}
\end{equation*}
\end{problem}
We denote by $\transportset(\cellset)$ the feasible set of Problem~\ref{problem:traffic_optimistic}, and we refer to a feasible solution $\transport(\cellset) := [\alpha_{\ivar\jvar}]$ as a \emph{transportation matrix}.
We denote by $\underline\transport(\cellset) := [\underline\alpha_{\ivar\jvar}]$ the \emph{optimal} solution of Problem~\ref{problem:traffic_optimistic}, which we refer to as the \emph{optimistic matrix} of partition $\cellset$, and we denote by $\underline\avglink(\cellset)$ the cost of the optimal solution.

\begin{lemma}[Lower bound on the cost of EBMP]
\label{lemma:avgmatch_weaklowerbound}
Let $\pickset_\numdem,\delvset_\numdem \sim \ESCPinst$,
and let $Q_\numdem = (\pickset_\numdem,\delvset_\numdem)$.
For any finite partition $\cellset$ of $\env$,
$\liminf_{\numdem\to\infty} l_\text{M}(Q_\numdem) \geq \underline \avglink(\cellset)$
almost surely, 
where $\avglink_\text{M}(Q_\numdem)$ is the average length of a match in the optimal bipartite matching, and $\underline\avglink(\cellset)$ denotes the value of Problem~\ref{problem:traffic_optimistic}.
\end{lemma}

\ifSHORTversion
\begin{proof}
Let $\permvar$ denote the optimal bipartite matching of $Q_\numdem$.
For a particular partition $\cellset$, we define random variables
$\hat\alpha_{\ivar\jvar} := \left| \left\{ \kvar : \Yvar_\kvar \in \workcell^\ivar, \Xvar_{\permelem{\kvar}} \in \workcell^\jvar \right\} \right| / \numdem$
for every pair $(\workcell^\ivar,\workcell^\jvar)$ of cells;
that is, $\hat\alpha_{\ivar\jvar}$ denotes the fraction of matches under $\permvar$ whose $\delvset$-endpoints are in $\workcell^\ivar$ and whose $\pickset$-endpoints are in $\workcell^\jvar$.
Let $\hat\transportset_\numdem$ be the set of matrices with entries $\left\{ \alpha_{\ivar\jvar} \geq 0 \right\}_{\ivar,\jvar = 1,\ldots,|\cellset|}$,
such that $\sum_\ivar \alpha_{\ivar\jvar} = \left| \pickset_\numdem \cap \workcell^\jvar \right| / \numdem$ for all $\workcell^\jvar \in \cellset$
and $\sum_\jvar \alpha_{\ivar\jvar} = \left| \delvset_\numdem \cap \workcell^\ivar \right| / \numdem$ for all $\workcell^\ivar \in \cellset$;
note $\left\{ \hat\alpha_{\ivar\jvar} \right\}$ itself is an element of $\hat\transportset_\numdem$.
Then the average match length $\avglink_\text{M}(Q_\numdem)$ is bounded below by
\begin{equation*}
\begin{aligned}
\avglink_\text{M}(Q_\numdem) =
\oneover{\numdem} \sum_{\kvar=1}^\numdem \vecnorm{ \Xvar_{\sigma(\kvar)} - \Yvar_\kvar }
  & \geq \sum_{\ivar\jvar} \hat\alpha_{\ivar\jvar} \min_{y\in\workcell^\ivar, x\in\workcell^\jvar} \vecnorm{x-y}	\\
  & \geq \min_{A \in \hat\transportset_\numdem} \sum_{\ivar\jvar} \alpha_{\ivar\jvar} \min_{y\in\workcell^\ivar, x\in\workcell^\jvar} \vecnorm{x-y}.
\end{aligned}
\end{equation*}
The key observation is that
$\lim_{\numdem\to\infty} \left| \left\{ \pickset_\numdem \cap \workcell^\jvar \right\} \right| / \numdem = \pdistr_\picktag( \workcell^\jvar )$,
and $\lim_{\numdem\to\infty} \left| \left\{ \delvset_\numdem \cap \workcell^\ivar \right\} \right| / \numdem = \pdistr_\delvtag( \workcell^\ivar )$, almost surely.
Applying standard sensitivity analysis (see Chapter~5 of~\cite{bertsimas1997introduction}), we observe that the final expression converges almost surely to $\underline\avglink(\cellset)$
as $n \to +\infty$.
A detailed analysis of this type is included in the arXiv version of the paper.
\end{proof}

\else
\begin{proof}
See Appendix~\ref{appendix:lemmas} for the proof.
\end{proof}
\fi

We are interested in the tightest possible lower bound, and so we define
$\underline\avglink := \sup_\cellset \underline\avglink(\cellset)$.
Remarkably, the supremum lower bound $\underline\avglink$ is equivalent to the Wasserstein distance between $\pdistr_\delvtag$ and $\pdistr_\picktag$, and so we can refine Lemma~\ref{lemma:avgmatch_weaklowerbound} as follows.
\begin{lemma}[Best lower bound on the cost of EBMP]
\label{lemma:avgmatch_lowerbound}
Let $\pickset_\numdem,\delvset_\numdem \sim \ESCPinst$,
and let $Q_\numdem = (\pickset_\numdem,\delvset_\numdem)$.
Then
\begin{equation}
 \label{eq:supremumbound}
\liminf_{\numdem\to\infty} l_\text{M}(Q_\numdem) \geq W(\pdistr_\delvtag,\pdistr_\picktag),
\qquad \text{almost surely.}
\end{equation}
\end{lemma}
\ifSHORTversion
\begin{proof}
The lemma is proved by showing that
$\sup_\cellset \underline\avglink(\cellset) = W(\pdistr_\delvtag,\pdistr_\picktag)$.
By construction, Problem~\ref{problem:traffic_optimistic} is a discrete approximation (and lower bound) of~\eqref{eq:wasserstein};
moreover, it can be shown that $\lim_{\sidecuts \to +\infty} \underline\avglink(\cellset_\sidecuts) - W(\pdistr_\delvtag,\pdistr_\picktag) \to 0^-$, where $\cellset_\sidecuts$ is the grid partition of $\probdim^\sidecuts$ cubes.
Applying Lemma~\ref{lemma:avgmatch_weaklowerbound} to this sequence of partitions obtains the lemma.
For a more detailed proof we refer to the arXiv version of the paper.
\end{proof}
\else
\begin{proof}
The lemma is proved by showing that
$\sup_\cellset \underline\avglink(\cellset) = W(\pdistr_\delvtag,\pdistr_\picktag)$.
See Appendix~\ref{appendix:lemmas}.
\end{proof}
\fi

Henceforth in the paper, we will abandon the notation $\underline \avglink$ in favor of $W(\pdistr_\delvtag,\pdistr_\picktag)$ to denote this lower bound.
This connection to the Wasserstein distance yields the following perhaps surprising result.
\begin{prop} \label{prop:linkpositive}
The supremum lower bound $W(\pdistr_\delvtag,\pdistr_\picktag)$
of~\eqref{eq:supremumbound} is equal to zero
if and only if $\pdistr_\delvtag = \pdistr_\picktag$.
\end{prop}
\ifSHORTversion
\begin{proof}
The proposition follows immediately from the fact that the Wasserstein distance is known to satisfy the axioms of a metric on $\Gamma(\pdistr_\delvtag,\pdistr_\picktag)$.
For an alternative proof, we refer to the arXiv version of the paper.
\end{proof}
\else
\begin{proof}
The proposition follows immediately from the fact that the Wasserstein distance is known to satisfy the axioms of a metric on $\Gamma(\pdistr_\delvtag,\pdistr_\picktag)$.
Nevertheless, we provide a short alternative proof.

The proof of the forward direction is by construction:
Suppose $\pdistr_1 = \pdistr_2 = \pdistr$;
let $\gamma$ be the measure such that
$\gamma( J ) = \int_{ x : (x,x) \in J } \pdistr(x) dx$
for any $J \in \env \times \env$;
clearly $\gamma \in \Gamma(\pdistr,\pdistr)$.
Then $\int_{x,y \in \env} \vecnorm{y-x} d\gamma(x,y) = \int_{x \in \env} \vecnorm{x-x} \pdistr(x) dx = 0$.

The proof of the reverse direction is by contradiction:
Suppose $\pdistr_1 \neq \pdistr_2$.
Then one can choose $\epsilon > 0$ sufficiently small, and regions $A_1$ and $A_2$ (where $A_1 \subseteq A_2 \subseteq \env$), so that $\pdistr_1(A_1) > \pdistr_2(A_2)$
and $\vecnorm{x-y} \geq \epsilon$ for all $x \in A_1$ and $y \notin A_2$.
Then for any $\gamma \in \Gamma(\pdistr_1,\pdistr_2)$
\[
\begin{aligned}
\int_{x,y \in \env} \vecnorm{x-y} d\gamma(x,y)
  &\geq \epsilon \gamma\left( \{ (x,y) : x \in A_1, y \notin A_2 \} \right)	\\
  &= \epsilon \left[ \pdistr_1(A_1) - \gamma\left( \{ (x,y) : x \in A_1, y \in A_2 \} \right) \right]	\\
  &\geq \epsilon \left[ \pdistr_1(A_1) - \pdistr_2(A_2) \right] > 0.
\end{aligned}
\]
\end{proof}
\fi
The intuition behind the alternative proof is that if some fixed area $\mathcal A$ in the environment has unequal proportions of $\pickset$ points versus $\delvset$ points, then a positive fraction of the matches associated with $\mathcal A$ (a positive fraction of all matches) must have endpoints \emph{outside} of $\mathcal A$, i.e., at positive distance. Such an area can be identified whenever $\pdistr_\picktag \neq \pdistr_\delvtag$.

Thus the implication of Lemma~\ref{lemma:avgmatch_lowerbound} is that the average match length is asymptotically no less than some constant which depends \emph{only} on the workspace geometry and the spatial distribution of pickup and delivery points; moreover, that constant is generally \emph{non-zero}.
We are now in a position to state the main result of this section.
\begin{theorem}[Lower bound on the cost of ESCP]\label{thrm:lowb}
Let $\tourlen^*(n)$ be the length of the optimal stacker crane tour through
  $\pickset_\numdem,\delvset_\numdem \sim \ESCPinst$, for compact $\env \in \reals^\probdim$, where $d\geq 2$.
Then
\begin{equation} \label{eq:scpLB}
  \liminf_{\numdem\to+\infty} \tourlen^* (n)/ \numdem
  \geq
  \expect_{\spatialjoint} \vecnorm{Y-X} + W(\pdistr_\delvtag,\pdistr_\picktag),
\end{equation}
almost surely.
\end{theorem}
\begin{proof}
A stacker crane tour is composed of pickup-to-delivery links and delivery-to-pickup links. The latter describe \emph{some} bipartite matching having cost no less than the optimal cost for the EBMP.
Thus, one can write
\begin{align*}
  \tourlen^*(n) / \numdem
  &\geq
    \frac{1}{\numdem} \sum_{i=1}^\numdem \vecnorm{\Yvar_i-\Xvar_i}
    + \frac{1}{\numdem} \M(Q_n).
\end{align*}
By the Strong Law of Large Numbers, the first term of the last expression goes to $\expect_{\spatialjoint}\vecnorm{Y-X}$ almost surely.
By Lemma~\ref{lemma:avgmatch_lowerbound}, the second term is bounded below asymptotically, almost surely, by $W(\pdistr_\delvtag,\pdistr_\picktag)$.
\end{proof}

\begin{remark}[Lower bound on the cost of the $m$ESCP]\label{rem:mETSP}
The multi-vehicle ESCP ($m$ESCP) consists in finding a \emph{set} of $m$ stacker crane tours such that all pickup-delivery pairs are visited exactly once and the total cost is minimized.
The $m$ECSP arises when more than one vehicle is available for service.
It is straightforward to show that the lower bound in Theorem \ref{thrm:lowb} is also a valid lower bound for the optimal cost of the $m$ESCP, for any $m$.
\end{remark}

\subsection{An Upper Bound on the Length of the ESCP}\label{subsec:linkupper}

In this section we produce a sequence that bounds $\tourlen_\text{M}(Q_\numdem)$ asymptotically from above, and matches the linear scaling of~\eqref{eq:scpLB}.
The bound relies on the performance of Algorithm~\ref{alg:randEBMP}, a \emph{randomized} algorithm for the stochastic EBMP.
The idea of Algorithm~\ref{alg:randEBMP} is that each point $y \in \delvset$ randomly generates an associated \emph{shadow site} $X'$,
so that the collection $\pickset'$ of shadow sites ``looks like'' the set of actual pickup sites.
An \emph{optimal} matching is produced between $\pickset'$ and $\pickset$
which assists in the matching between $\delvset$ and $\pickset$;
specifically, if $x \in \pickset$ is the point matched to $X'$, then the matching produced by Algorithm~\ref{alg:randEBMP} contains $(y,x)$.
An illustrative diagram can be found in Figure~\ref{fig:shadowmatch}.

Algorithm~\ref{alg:randEBMP} is specifically designed to have two important properties for random sets $Q_n$:
First, $\expect\vecnorm{X'-Y}$ is predictably controlled by ``tuning" inputs---a partition $\cellset$ of the environment and ``policy matrix'' $\transport(\cellset)$---
chosen as a function of $n$;
second, $\tourlen_\text{M}( (\pickset',\pickset) )/\numdem \to 0^+$ as $\numdem\to+\infty$.
Later we will show that $\cellset$ and $\transport(\cellset)$ can be chosen so that
$\expect\vecnorm{X'-Y} \to W(\pdistr_\delvtag,\pdistr_\picktag)$ (as $\numdem \to +\infty$)
leading to a bipartite matching algorithm whose performance matches the lower bound of~\eqref{alg:randEBMP}.
\begin{algorithm}[htb]
\caption{Randomized EBMP (parameterized)}
\label{alg:randEBMP}
\begin{algorithmic}[1]
\renewcommand{\algorithmicrequire}{\textbf{Input:}}
\renewcommand{\algorithmicensure}{\textbf{Output:}}
\renewcommand{\algorithmiccomment}[1]{ // \emph{#1} }
\REQUIRE
\emph{pickup} points $\pickset = \{\xvar_1,\ldots,\xvar_\numdem\}$,
\emph{delivery} points $\delvset = \{\yvar_1,\ldots,\yvar_\numdem\}$, probability densities $\pdistr_\picktag(\cdot)$ and $\pdistr_\delvtag(\cdot)$,
partition $\cellset$ of the workspace,
and matrix $\transport(\cellset) \in \transportset(\cellset)$.	
\ENSURE a bi-partite matching between $\delvset$ and $\pickset$.
\STATE {\bf initialize} $\pickset' \gets \emptyset$.
\STATE {\bf initialize} matchings
$\overline M \gets \emptyset$;
$\hat M \gets \emptyset$;
$M \gets \emptyset$.
\STATE \COMMENT{generate ``shadow pickups''}	
\FOR{$y \in \delvset$}
\STATE Let $\workcell^i$ be the cell containing $y$.
\label{line:shadowsample_start}
\STATE Sample $j$; $j = \jvar'$ with probability
  $\alpha_{\ivar\jvar'} / \pdistr_\delvtag(\workcell^\ivar)$.
\STATE Sample $X'$ with pdf
  $\pdistr_\picktag( \ \cdot \ | \Xvar' \in \workcell^\jvar )$.
\label{line:shadowsample_stop}
\STATE Insert $X'$ into $\pickset'$ and $(y,X')$ into $\overline M$.
\ENDFOR
%
\STATE $\hat M \gets$ an optimal EBMP between $\pickset'$ and $\pickset$.
\STATE \COMMENT{construct the matching}
\FOR{$X' \in \pickset'$}
\STATE Let $(y,X')$ and $(X',x)$ be the matches in $\overline M$ and $\hat M$, respectively, whose $\pickset'$-endpoints are $X'$.
\STATE Insert $(y,x)$ into $M$.
\ENDFOR

\RETURN $M$	
\end{algorithmic}
\end{algorithm}

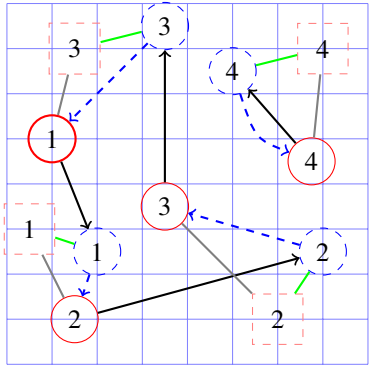
\begin{figure}[tb]
\centering
\definecolor{gold}{rgb}{.85,.64,.125}

\begin{tikzpicture}[scale=.6]
\tikzstyle{pickup}=[circle,draw=red!100] 
\tikzstyle{delivr}=[circle,dashed,draw=blue!100] 
\tikzstyle{shadow}=[regular polygon,regular polygon sides=4,dashed,draw=red!50]

\tikzstyle{p2d}=[->,thick]
\tikzstyle{stacker}=[->,thick,dashed,color=gold]
\tikzstyle{match}=[->,thick,dashed,color=blue]
\tikzstyle{primary}=[color=green,thick]
\tikzstyle{node-done} = [circle,draw=gray!100] 
\tikzstyle{edge-done} = [color=gray,thick]

\draw[step=1,color=blue!50,very thin] (-4,-4) grid (4,4) ;		

\coordinate (p1) at (-3,1) ;		\coordinate (d1) at (-2, -1.5) ;	\coordinate (s1) at (-3.5,-1) ;
\coordinate (p2) at (-2.5,-3) ;		\coordinate (d2) at (3,-1.5) ;		\coordinate (s2) at (2,-3) ;
\coordinate (p3) at (-.5,-.5) ;		\coordinate (d3) at (-.5,3.5) ;		\coordinate (s3) at (-2.5,3) ;
\coordinate (p4) at (2.75,.5) ;		\coordinate (d4) at (1,2.5) ;		\coordinate (s4) at (3,3) ;


\node [delivr] (delv1) at (d1) {1} ;
\node [pickup,thick] (pick1) at (p1) {1} ; \draw [p2d] (pick1) -- (delv1) ;
\node [delivr] (delv2) at (d2) {2} ;
\node [pickup] (pick2) at (p2) {2} ; \draw [p2d] (pick2) -- (delv2) ;
\node [delivr] (delv3) at (d3) {3} ;
\node [pickup] (pick3) at (p3) {3} ; \draw [p2d] (pick3) -- (delv3) ;
\node [delivr] (delv4) at (d4) {4} ;
\node [pickup] (pick4) at (p4) {4} ; \draw [p2d] (pick4) -- (delv4) ;

\node [shadow] (shad1) at (s1) {1} ;	\draw [primary] 	(delv1) -- (shad1) ;
\node [shadow] (shad2) at (s2) {2} ;	\draw [primary] 	(delv2) -- (shad2) ;
\node [shadow] (shad3) at (s3) {3} ;	\draw [primary] 	(delv3) -- (shad3) ;
\node [shadow] (shad4) at (s4) {4} ;	\draw [primary] 	(delv4) -- (shad4) ;

\draw [edge-done] 	(shad1) -- (pick2) ;
\draw [match]		(delv1) -- (pick2) ;
\draw [edge-done] 	(shad2) -- (pick3) ;
\draw [match]		(delv2) -- (pick3) ;
\draw [edge-done] 	(shad3) -- (pick1) ;
\draw [match]		(delv3) -- (pick1) ;
\draw [edge-done] 	(shad4) -- (pick4) ;
\draw [match]		(delv4) .. controls (1.5,1) .. (pick4) ;


\end{tikzpicture}
\caption{
Algorithm~\ref{alg:randEBMP}:
Demands are labeled with integers.
Pickup and delivery sites are represented by solid and dashed circles, respectively.
Pickup-to-delivery links are shown as black arrows.
Shadow pickups are shown as dashed squares, with undirected links to their generators (delivery sites);
also shown are optimal matching links between shadows and pickups.
Dashed arrows show the resulting induced matching.
Note, this solution produces two disconnected subtours $(1,2,3)$ and $(4)$.
}
\label{fig:shadowmatch}
\end{figure}

We present the first two properties as formal lemmas:
\begin{lemma}[Similarity of $\pickset'$ to $\pickset$]
\label{lemma:shadows}
Let $X_1,\ldots,X_\numdem$ be a set of points that are i.i.d. with density $\pdistr_\picktag$;
let $Y_1,\ldots,Y_\numdem$ be a set of points that are i.i.d. with density  $\pdistr_\delvtag$.
Then Algorithm~\ref{alg:randEBMP} generates shadow sites $X'_1,\ldots,X'_\numdem$,
which are (i) jointly independent of $X_1,\ldots,X_\numdem$,
and (ii) mutually i.i.d.,
with density $\pdistr_\picktag$.
\end{lemma}
\smallskip
\ifSHORTversion
\begin{proof}
Lemma~\ref{lemma:shadows} relies on basic laws of probability, and the proof is omitted in the interest of brevity.
For a complete proof we refer to the arXiv version of the paper.
\end{proof}
\else
\begin{proof}
Lemma~\ref{lemma:shadows} relies on basic laws of probability, and its proof is relegated to Appendix~\ref{appendix:lemmas}. 
\end{proof}
\fi
The importance of this lemma is that it allows us to apply equation~\eqref{eq:EBMP_common} of Section~\ref{subsec:EBMP} to characterize $\M((\pickset',\pickset))$.
\begin{lemma}[Delivery-to-Pickup Lengths]
\label{lemma:sampletravel}
Let $Y$ be a random point with probability density $\pdistr_\delvtag$;
let $X'$ be the \emph{shadow site} of $y=Y$ generated by lines~\ref{line:shadowsample_start}-\ref{line:shadowsample_stop} of Algorithm~\ref{alg:randEBMP},
running with inputs $\cellset$ and $\transport(\cellset)$.
Then
\begin{align*}
  \expect \vecnorm{X'-Y}
  \leq \sum_{\ivar\jvar}
    \alpha_{\ivar\jvar} \max_{\yvar \in \workcell^\ivar, \xvar \in \workcell^\jvar}
    \vecnorm{ \xvar - \yvar }.
\end{align*}
\end{lemma}

\ifSHORTversion
\begin{proof}
The proof...
\end{proof}

\else
\begin{proof}
Again, see Appendix~\ref{appendix:lemmas}.
\end{proof}
\fi
\smallskip

Given a finite partition $\cellset$, it should be desirable to choose $\transport(\cellset)$ in order to optimize the performance of Algorithm~\ref{alg:randEBMP};
that is, minimize the expected length of the matching produced.
We can minimize at least the bound of Lemma~\ref{lemma:sampletravel} using the solution of Problem~\ref{problem:traffic_pessimistic}, shown below.
We refer to its solution $\overline\transport(\cellset)$ as the \emph{pessimistic matrix} of partition $\cellset$.
\begin{problem}[Pessimistic ``rebalancing'']
\label{problem:traffic_pessimistic}
\begin{equation*}
\begin{aligned}
	& \underset{
	  \left\{ \alpha_{\ivar\jvar} \geq 0 \right\}_{ i,j \in \{1,\ldots,r^d\} } }{\text{Minimize}}
	& & \sum_{\ivar\jvar} \alpha_{\ivar\jvar} \max_{y\in\workcell^\ivar, x\in\workcell^\jvar} \vecnorm{x-y} \\
	& \text{subject to}
	& &		\sum_\jvar \alpha_{\ivar\jvar} = \pdistr_\delvtag( \workcell^\ivar )
			\quad \text{for all $\workcell^\ivar \in \cellset$,}
	\\
	&&&		\sum_\ivar \alpha_{\ivar\jvar} = \pdistr_\picktag( \workcell^\jvar )
			\quad \text{for all $\workcell^\jvar \in \cellset$.}
\end{aligned}
\end{equation*}
\end{problem}

Now we present Algorithm~\ref{alg:randEBMP2}, described in pseudo-code, which computes%
\footnote{In the definition of the algorithm, we use the ``small omega'' notation, where $f(\cdot) \in \omega( g(\cdot) )$ implies $\lim_{n\to\infty} f(n)/g(n) = \infty$.}
a specific partition $\cellset$,
and then invokes Algorithm~\ref{alg:randEBMP} with inputs $\cellset$ and $\overline\transport(\cellset)$.
\begin{algorithm}[htb]
\caption{Randomized EBMP}
\label{alg:randEBMP2}
\begin{algorithmic}[1]
\renewcommand{\algorithmicrequire}{\textbf{Input:}}
\renewcommand{\algorithmicensure}{\textbf{Output:}}
\renewcommand{\algorithmiccomment}[1]{ // \emph{#1} }
\REQUIRE
\emph{pickup} points $\pickset = \{\xvar_1,\ldots,\xvar_\numdem\}$,
\emph{delivery} points $\delvset = \{\yvar_1,\ldots,\yvar_\numdem\}$,
and prob. densities $\pdistr_\picktag(\cdot)$ and $\pdistr_\delvtag(\cdot)$.
\ENSURE a bi-partite matching between $\delvset$ and $\pickset$.
\renewcommand{\algorithmicrequire}{\textbf{Require:}}
\REQUIRE an arbitrary \emph{resolution} function
$\text{res}(\numdem) \in \omega(\numdem^{1/\probdim})$, where $\probdim$ is the dimension of the space.
\STATE $\sidecuts \gets \text{res}(\numdem)$.
\STATE $\cellset \gets$ grid partition $\cellset_\sidecuts$, of $\sidecuts^\probdim$ cubes.
\STATE $\transport \gets \overline\transport(\cellset)$,
the solution of Problem~\ref{problem:traffic_pessimistic}.
\STATE Run Algorithm~\ref{alg:randEBMP} on $\left( \pickset,\delvset,\pdistr_\picktag,\pdistr_\delvtag,\cellset,\transport \right)$,
	producing matching $M$.
\RETURN $M$
\end{algorithmic}
\end{algorithm}

\begin{lemma}[Granularity of Algorithm~\ref{alg:randEBMP2}]
\label{lemma:diffB}
Let $\sidecuts$ be the resolution parameter, and $\cellset_\sidecuts$ the resulting grid-based partition, used by Algorithm~\ref{alg:randEBMP2}.
Let $Y$ be a random variable with probability density $\pdistr_\delvtag$,
and let $X'$ be the \emph{shadow site} of $y = Y$ generated by lines~\ref{line:shadowsample_start}-\ref{line:shadowsample_stop} of Algorithm~\ref{alg:randEBMP},
running under Algorithm~\ref{alg:randEBMP2}.
Then
$\expect \vecnorm{X'-Y} - W(\pdistr_\delvtag,\pdistr_\picktag)
  \leq 2\sidelength \sqrt{\probdim} / \sidecuts$.
\end{lemma}
\begin{proof}
See Appendix~\ref{appendix:lemmas} for the proof.
\end{proof}

We are now in a position to present an upper bound on the cost of the optimal EBMP matching that holds in the general case when $\den_{\text{P}}\neq \den_{\text{D}}$.

\begin{lemma}[Upper bound on the cost of EBMP]
\label{lemma:ebmp_upper}
Let $\pickset_\numdem,\delvset_\numdem \sim \ESCPinst$,
and let $Q_\numdem = (\pickset_\numdem,\delvset_\numdem)$.
For $d\geq 3$,
\begin{equation} \label{eq:BPMdistinct}
\begin{aligned}
  & \limsup_{\numdem\to+\infty} 
      \frac{
	\tourlen_\text{M}(Q_\numdem) - \numdem W(\pdistr_\delvtag,\pdistr_\picktag)
      }{ \numdem^{1-1/\probdim} }
  \leq \kappa(\pdistr_\picktag,\pdistr_\delvtag),
\end{aligned}
\end{equation}
almost surely,
where
\begin{equation} \label{eq:residualconst}
\kappa(\pdistr_\picktag,\pdistr_\delvtag) :=
  \min_{ \phi \in \{ \pdistr_\picktag, \pdistr_\delvtag \} } \left\{
    \beta_{\text{M},\probdim} \int_\env \phi(\locvar)^{1-1/\probdim} \ d\locvar
  \right\}.
\end{equation}
For $d=2$,
\begin{equation} \label{eq:BPMdistinctd2}
  \frac{ \tourlen_\text{M}(Q_\numdem)
    - \numdem W(\pdistr_\delvtag,\pdistr_\picktag) }{\sqrt{n\, \log n}}
  \leq
  \gamma,
\end{equation}
with high probability as $n \to +\infty$,
for some positive constant $\gamma$.
\end{lemma}

\begin{proof}
See Appendix~\ref{appendix:lemmas} for the proof.
\end{proof}

We can leverage this result to derive the main result of this section, which is an asymptotic upper bound for the optimal cost of the ESCP. In addition to having the same linear scaling as our lower bound, the bound also includes ``next-order'' terms.
\begin{theorem}[Upper bound on the cost of ESCP]
\label{thm:scp_upper}
Let $\pickset_\numdem,\delvset_\numdem \sim \ESCPinst$ be a random instance of the ESCP, for compact $\env \in \reals^\probdim$, where $d\geq 2$.
Let $\tourlen^*(n)$ be the length of the optimal stacker crane tour through $\pickset_\numdem \cup \delvset_\numdem$.
Then, for $d\geq 3$, 
\begin{equation} \label{eq:spliceUB}
\begin{aligned}
  & \limsup_{\numdem\to+\infty} 
      \frac{
	\tourlen^*(n)
	- \numdem \Bigl[ \expect_\spatialjoint\vecnorm{Y-X} + W(\pdistr_\delvtag,\pdistr_\picktag) \Bigr]
      }{ \numdem^{1-1/\probdim} }
  \leq \kappa(\pdistr_\picktag,\pdistr_\delvtag),
\end{aligned}
\end{equation}
almost surely.
For $d=2$,
\begin{equation} \label{eq:stacker_ub_d2}
  \frac{
\tourlen^*(n)
    - \numdem \Bigl[ \expect_\spatialjoint\vecnorm{Y-X} + W(\pdistr_\delvtag,\pdistr_\picktag) \Bigr]
  }{ \sqrt{n\, \log n} }
  \leq
  \gamma,
\end{equation}
with high probability as $n \to +\infty$, for some positive constant $\gamma$.
\end{theorem}

\begin{proof}
We first consider the case $d\geq 3$. 
Let $\tourlen_\text{SPLICE}(n)$ be the length of the SCP tour through $\pickset_\numdem,\delvset_\numdem$ generated by \ref{alg:splicing}. Let $Q_\numdem = (\pickset_\numdem,\delvset_\numdem)$.
One can write
\begin{align*}
  \tourlen_\text{SPLICE}(n) 
  &\leq
  \sum_{i=1}^n \vecnorm{\Yvar_i-\Xvar_i}
    \! + \! \M(Q_\numdem)
    + \max_{x,y\in \env}\|x-y\|\, \subn_\batchsize	\\
  &=
  \left( \sum_{i=1}^n \vecnorm{\Yvar_i-\Xvar_i} - \batchsize\expect_\spatialjoint\vecnorm{\Yvar-\Xvar} \right)	
      + \Bigl( \M(Q_\numdem) - \numdem W(\den_\delvtag,\den_\picktag) \Bigr)	\\
      & \qquad\qquad {} + \batchsize \Bigl[ \expect_\spatialjoint\vecnorm{\Yvar-\Xvar} + W(\den_\delvtag,\den_\picktag) \Bigr]
      + \max_{x,y\in \env}\vecnorm{x-y}\, \subn_\batchsize.
\end{align*}
The following results hold almost surely: The first term of the last expression is $o(\batchsize^{1-1/\probdim})$ (absolute differences); by Lemma~\ref{lemma:ebmp_upper}, the second term is $\kappa(\pdistr_\picktag,\pdistr_\delvtag) \batchsize^{1-1/\probdim} + o(\batchsize^{1-1/\probdim})$; finally, by Remark~\ref{rem:van}, one has $\lim_{n \to +\infty} \subn_\batchsize/n^{1-1/d}=0$.
Collecting these results, dividing on both sides by $n^{1-1/d}$, and noting that by definition $\tourlen^*(n)\leq  \tourlen_\text{SPLICE}(n) $, one obtains the claim. The proof for the case $d=2$ is almost identical and is omitted.

\end{proof}

\section{Stability Condition for DRT Systems} \label{sec:dynamicversion}

In the previous section we presented new asymptotic results for the length of the stochastic EBMP and ESCP.
We showed convergence to linearity in the size $\batchsize$ of the instance, and characterized \emph{next-order} growth as well (equation~\eqref{eq:spliceUB} and equation \eqref{eq:stacker_ub_d2}).
Here we use such new results to derive a necessary and sufficient condition for the stability of DRT systems, modeled as DPDPs.

Let us define the \emph{load factor} as
\begin{equation}\label{eq:rho}
\utilization := \lambda \, [ \expect_\spatialjoint\vecnorm{Y-X} + W(\pdistr_\delvtag,\pdistr_\picktag)]/m.
\end{equation}
Note that when $\pdistr_\delvtag = \pdistr_\picktag$, one has $W(\pdistr_\delvtag,\pdistr_\picktag) = 0$ (by Proposition~\ref{prop:linkpositive}), and the above definition reduces to the definition of load factor given in \cite{Treleaven.Pavone.ea:10} (valid for $d\geq 3$ and $\pdistr_\delvtag = \pdistr_\picktag$).
The following theorem gives a necessary and sufficient condition for a stabilizing routing policy to exist.
\begin{theorem}[Stability condition for DRT systems]\label{thrm:main}
Consider the DPDP defined in Section~\ref{sec:problem}, which serves as a model of DRT systems. Then, the condition $\utilization < 1$
is necessary and sufficient for the existence of stabilizing policies.
\end{theorem}
\smallskip

\begin{proof}[Proof of Theorem~\ref{thrm:main} --- Part I: Necessity]

Consider any causal, stable routing policy (since the policy is stable and the arrival process is Poisson, the system has renewals and the inter-renewal intervals are finite with probability one). Let  $A(t)$ be the number of demand arrivals from time $0$ (when the first arrival occurs) to time $t$. Let $R(t)$ be the number of demands in the process of receiving service at time $t$ (a demand is in the process of  receiving service if a vehicle is traveling toward its pickup location or a vehicle is transporting such demand to its delivery location). Finally, let $S_i$ be the servicing time of the $i$th demand  (this is the time spent by a vehicle to travel to the demand's pickup location and to transport such demand to its delivery location).

The time average number of  demands in the process of receiving service is given by
\[
\bar \utilization: = \lim_{t\to +\infty} \, \frac{1}{t} \int_{\tau=0}^{t}\, R(\tau)\, d\tau.
\]
By following the arguments in \cite[page 81-85, Little's Theorem]{Gallager:96}, $\bar \rho$ can be written as:
\[
\bar \utilization = \lim_{t \to +\infty} \, \frac{\sum_{i=1}^{A(t)}\, S_i}{t} = \lim_{t \to +\infty} \,  \frac{\sum_{i=1}^{A(t)}\, S_i}{A(t)} \, \lim_{t \to +\infty}\,  \frac{A(t)}{t},  
\]
where the first equality holds almost surely and all limits exist almost surely. The second limit on the right is, by definition, the arrival rate $\lambda$. The first limit on the right can be lower bounded as follows:
\[
\lim_{t \to +\infty} \,  \frac{\sum_{i=1}^{A(t)}\, S_i}{A(t)} \geq \lim_{t \to +\infty} \,  \frac{\tourlen^*(A(t))}{A(t)},
\]
where $\tourlen^*(A(t))$ is the optimal length of the \emph{multi-vehicle} stacker crane tour through the $A(t)$ demands (i.e., is the optimal solution to the multi-vehicle ESCP - see Remark \ref{rem:mETSP}-). For any sample function, $\tourlen^*(A(t))/A(t)$ runs through the same sequence of values with increasing $t$ as $\tourlen^*(n)/n$ runs through with increasing $n$. Hence, by Theorem \ref{thrm:lowb} and Remark \ref{rem:mETSP} we can write, almost surely,
\[
 \lim_{t \to +\infty} \,  \frac{\tourlen^*(A(t))}{A(t)} \geq \expect_\spatialjoint\vecnorm{Y-X} + W(\pdistr_\delvtag,\pdistr_\picktag).
\]

Collecting the above results, we obtain:
\[
\bar \varrho \geq \lambda \, [ \expect_\spatialjoint\vecnorm{Y-X} + W(\pdistr_\delvtag,\pdistr_\picktag)],
\]
almost surely. Since the policy is stable and the arrivals are Poisson, the \emph{per-vehicle} time average number of demands in the process of receiving service must be strictly less than one, i.e.,  $\bar \utilization/m <1$; this implies that for any causal, stable routing policy
\[
\lambda \, [ \expect_\spatialjoint\vecnorm{Y-X} + W(\pdistr_\delvtag,\pdistr_\picktag)]/m <1,
\] 
and necessity is proven.

\end{proof}
\begin{proof}[Proof of Theorem~\ref{thrm:main}---Part II: Sufficiency ]
The proof of sufficiency is constructive in the sense that we design a particular policy that is stabilizing.
In particular, a \emph{gated} policy is stabilizing which performs the following steps any time all servers are idle:
(1) applies algorithm \ref{alg:splicing} to determine tours through the \emph{outstanding} demands, (2) splits the tour into $m$ equal length fragments, and (3) assigns a fragment to each vehicle.

Consider, first, the case $d=2$. Similarly as in the proof of Theorem 4.2 in \cite{Pavone.Frazzoli.ea.TAC11}, we derive a recursive relation bounding the expected number of demands in the system at the times when new tours are computed.
Specifically, let $t_i$, $i\geq 0$, be the time instant at which the $i$th SPLICE tour is constructed
(i.e. the previous servicing round was completed); we will call this instant \emph{epoch} $i$.
We refer to the time interval between epoch $i$ and epoch $i+1$ as the $i$th iteration;
let $n_i$ be the number of demands serviced during the $i$th iteration (i.e. all outstanding demands at its epoch),
and let $C_i$ be the interval duration.

\newcommand{\spliceiter}[1][i]{{ \tourlen_{\text{SPLICE}}(n_{#1}) }}

Demands arrive according to a Poisson process with rate $\lambda$, so we have
$\expectation{n_{i+1}} = \lambda \, \expectation{C_i}$ for all epochs.
The interval duration $C_i$ is equal to the time required to service the demands of the $i$th group of demands.
One can easily bound $\expectation{C_i}$ as
\begin{equation}\label{eq:epochduration}
\expectation{C_i} \leq
    \expectation{ \spliceiter } / m
    + \diam
  + \diam,
\end{equation}
where $\spliceiter$ denotes the length of the SPLICE tour through the $i$th iteration demands,
and $\diam \doteq \max \{ \|p-q\| \, |\, p,q \in \env\}$ is the diameter of $\env$;
the constant terms account conservatively for (i) extra fragment length incurred by splitting the tour between vehicles, and (ii) extra travel required to reach the current tour fragment from the endpoint of the previous fragment. Let  
\[
\psi(n):= n\,  \Bigl[ \expect_\spatialjoint\vecnorm{Y-X} + W(\pdistr_\delvtag,\pdistr_\picktag) \Bigr] + \gamma \sqrt{\numdem \log\numdem} + o\left( \sqrt{\numdem \log\numdem} \right),
\]
and let $q(n)$ be the probability that given $n$ demands equation~\eqref{eq:stacker_ub_d2}  does not hold, i.e.:
\[
\probNotHold(n) := 1 -  \mathbb{P} \biggl [ \tourlen_{\text{SPLICE}}(n) \leq  \psi(n)  \biggr].
\]

Now, let $\eps$ be an arbitrarily small positive constant. By Theorem~\ref{thm:scp_upper}, $\lim_{n\to +\infty}\, q(n) = 0$; hence, there exists a number $\bar n$ such that for all $n \geq \bar n$ one has $q(n)<\eps$.
\ifSHORTversion
Using the trivial upper bound $\tourlen_\text{SPLICE}(n) \leq 2\, n \, \diam  $ when equation~\eqref{eq:stacker_ub_d2} does not hold,
one can derive the following bound for the length of the optimal stacker crane tour through the $n_i$ demands.
\begin{equation*}
\expectation{ \spliceiter }
\leq
\expectation{ \psi(n_i)} + \bar{n}^2 \, \diam  + \eps \, 2 \, \diam \, \expectation{n_i}. 
\end{equation*}
\else
Then, the length of the optimal stacker crane tour through the $n_i$ demands can be upper bounded as follows (where in the third step we use the trivial upper bound $\tourlen_\text{SPLICE}(n) \leq 2\, n \, \diam  $, which is always valid):
\begin{equation*}
\begin{aligned}
\expectation{ \spliceiter }
&= \sum_{n=0}^{+\infty}\,  \condexpectation{ \spliceiter }{ n_i=n }  \, \prob{n_i = n}	\\
&= \sum_{n=0}^{+\infty}\,  \biggl[ \condexpectation{ \spliceiter  }{ \text{\eqref{eq:stacker_ub_d2} holds}, n_i=n }
  \, \underbrace{ \probcond{ \text{\eqref{eq:stacker_ub_d2} holds}}{n_i=n}}_{=  1-q(n) } +	\\
&\qquad \qquad\quad  \condexpectation{ \spliceiter }{ \text{\eqref{eq:stacker_ub_d2} does not hold}, n_i=n }
  \,
  \underbrace{ \probcond{ \text{\eqref{eq:stacker_ub_d2} does not hold}}{n_i=n}}_{= q(n)} \biggr]  \prob{n_i = n}	\\ 
& \leq  \sum_{n=0}^{+\infty}\,  \biggl[ \condexpectation{ \psi(n_i) }{ \text{\eqref{eq:stacker_ub_d2} holds}, n_i=n } \left( 1-q(n) \right) +	\\
&\qquad \qquad \quad \condexpectation{ 2\, n_i \, \diam }{ \text{\eqref{eq:stacker_ub_d2} does not hold}, n_i=n }\, q(n) \biggr]  \prob{n_i = n} \\
&=  \sum_{n=0}^{+\infty}\, \biggl[  \psi(n)  \left( 1-q(n) \right) + (2\, n \, \diam)\, q(n) \biggr]  \prob{n_i = n}  \\
& \leq \expectation{ \psi(n_i)} + \sum_{n=0}^{+\infty}\, ( 2\, n \, \diam  ) \, q(n) \,   \prob{n_i = n} \\
& = \expectation{ \psi(n_i)} + \sum_{n=0}^{\bar n-1}\, ( 2\, n \, \diam  ) \,  q(n) \, \prob{n_i = n}
  + \sum_{n=\bar n}^{+\infty}\, ( 2\, n \, \diam  ) \,  \underbrace{q(n)}_{<\eps} \, \prob{n_i = n} \\
& \leq  \expectation{ \psi(n_i)} + \bar{n}^2 \, \diam  + \eps \, 2 \, \diam \, \expectation{n_i}.
\end{aligned}
\end{equation*}
\fi

\ifSHORTversion
\else
\fi
Then, one can write the following recurrence relation:
\begin{equation*}
\begin{aligned}
\expectation{C_i}
&\leq
\expectation{ \spliceiter  }/m
+ 2\, \diam \\
& \leq \expectation{ \psi(n_i)  }/m + \eps \, 2 \, \diam \, \expectation{ n_i }/m
+ \diam (\bar{n}^2/m+2)\\
& \leq \expectation{n_i} 
    \, \bigl[ (\expect_\spatialjoint\vecnorm{Y-X} + W(\den_\delvtag,\den_\picktag)) + \eps\,2 \, \diam \bigr] / \numveh	\\
      &\qquad\qquad\quad 
	+ \expectation{ \gamma \sqrt{n_i \log n_i } +  o\left( \sqrt{\numdem_i \log\numdem_i} \right) } / \numveh
      + \diam (\bar{n}^2/m+2).
\end{aligned}
\end{equation*}

Now, let $\delta>0$; then\footnote{Consider, first, the case without the $ o(\sqrt{x \, \log x}) $ term. In this case, let $c(\delta) =  \frac{1}{2\delta}\Bigl( \log\frac{1}{\delta^2} - 1\Bigr) $. Then, by using Young's inequality, one can write: $ \frac{1}{2\delta}\Bigl( \log\frac{1}{\delta^2} - 1\Bigr) + \delta \, x - \sqrt{x\log x} \geq  \frac{1}{2\delta}\Bigl( \log\frac{1}{\delta^2} - 1\Bigr) + \delta \, x   - \frac{\delta\, x}{2} - \frac{\log x}{2\, \delta}=:\psi(x) $. Then, one can easily show that $\psi(x)\geq 0$ for all $x\geq 1$. The case with the $ o(\sqrt{x \, \log x}) $ term is similar and is omitted.} for all $x\geq 1$
\[
\sqrt{x \, \log x} + o(\sqrt{x \, \log x}) \leq  c(\delta)+ \delta\, x,
\]
where $c(\delta) \in \reals_{\geq 0}$ is a constant.
Hence, we can write the following recursive equation for trajectories  $i \mapsto \expectation{n_i} \in \reals_{\geq 0}$:
\begin{equation*}
\begin{split}
    \expectation{n_{i+1}}  = \lambda \expectation{C_i}\leq  \expectation{n_i} 
    \, \bigl[ \varrho  +\eps \,  \lambda \, 2 \, \diam / \numveh 
    +\lambda\,  \gamma\, \delta \,/ \numveh \bigr]+ \lambda\, \gamma \, c(\delta)/\numveh
    + \diam (\bar{n}^2/m+2).
\end{split}
\end{equation*}
Since $\varrho<1$ and $\eps$ and $\delta$ are arbitrarily small constants, one can \emph{always} find values for $\eps$ and $\delta$, say $\bar \eps$ and $\bar \delta$, such that 
\begin{equation*}
a(\utilization):= \varrho  +\bar \eps \,  \lambda \, 2 \, \diam / \numveh  +\lambda\,  \gamma\, \bar \delta \,/ \numveh<1.
\end{equation*} 

Hence, for trajectories  $i \mapsto \expectation{n_i} \in \reals_{\geq 0}$ we can write, for all $i\geq0$, a recursive upper bound:
\begin{equation}\label{eq:recDyn}
\expectation{n_{i+1}}  \leq a(\utilization)\,  \expectation{n_i} + \lambda\, \gamma \, c(\delta)/\numveh 
  + \diam (\bar{n}^2/m+2),
\end{equation}
with $a(\utilization)<1$, for \emph{any} $\utilization\in [0, \, 1)$.

We want to prove that trajectories $i \mapsto \expectation{n_i}$ are bounded, by studying the recursive upper bound in equation \eqref{eq:recDyn}. To this purpose we define an auxiliary system, System-X, whose trajectories $i \mapsto x_i \in \reals_{\geq 0}$ obey the dynamics:

\begin{equation}
x_{i+1}  =  a(\utilization) \, x_i + \lambda\, \gamma \, c(\delta)/\numveh + 2\, \diam (\bar{n}^2/m+1),
\end{equation}
with $x_0 = n_0$. By construction, trajectories $i \mapsto x_i $ upper bound trajectories $i \mapsto \expectation{n_i}$. One can easily note that  trajectories $i \mapsto x_i $ are indeed bounded for all initial conditions, since the eigenvalue of System-X, $a(\utilization)$, is strictly less than one. Hence,  trajectories $i \mapsto \expectation{n_i}$ are bounded as well and this concludes the proof for case $d=2$.

Case $d\geq3$ is virtually identical to case $d=2$, with the only exception that equation \eqref{eq:stacker_ub_d2}  should be replaced with equation \eqref{eq:spliceUB}, and the sublinear part  is given by $x^{1-1/d}$ (the fact that for $d\geq3$ the inequalities hold almost surely does not affect the reasoning behind the proof, since almost sure convergence implies convergence with high probability).
\end{proof}

Note that the stability condition in Theorem~\ref{thrm:main} depends on the workspace geometry, the stochastic distributions of pickup and delivery points, the demands' arrival rate, and the number of vehicles, and makes explicit the roles of the different parameters in affecting the performance of the overall system. We believe that this characterization would be instrumental for a system designer of DRT systems to build business and strategic planning models regarding, e.g., fleet sizing.

\begin{remark}[Load factor with non-unit velocity] 
In our model of DPDPs we have assumed, for simplicity, that vehicles travel at unit velocity. Indeed, Theorem \ref{thrm:main} holds also in the general case of vehicles with non-unit velocity $v$, with the only modification that the load factor is now given by
\[
\utilization := \frac{\lambda \, [ \expect_\spatialjoint\vecnorm{Y-X} + W(\pdistr_\delvtag,\pdistr_\picktag)]}{v \, m}.
\]
\end{remark}

\section{Simulation Results}\label{sec:sim}

In this section, we present simulation results to support each of the theoretical findings of the paper.
Our simulation experiments examine all three contributions of the paper; in particular, we discuss (i) performance of the \ref{alg:splicing}~algorithm, (ii) scaling of the length of the EBMP,
and (iii) stabilizability of the DPDP.

\subsection{Performance of SPLICE}

We begin the simulation section with a discussion of the performance of the \ref{alg:splicing}~algorithm.
We examine (i) the rate of convergence of \ref{alg:splicing} to the optimal solution,
(ii) the runtime for the \ref{alg:splicing} algorithm,
and (iii) the ratio between the runtime of \ref{alg:splicing} and that of an exact algorithm.
In all simulations we assume that the pickup/delivery pairs are generated i.i.d. in a unit cube and that $\pdistr_\picktag$ and $\pdistr_\delvtag$ are both uniform.
The Bipartite Matching Problem in line \ref{line:perm} of \ref{alg:splicing} is solved using the GNU Linear Programming Toolkit (GLPK) software on a linear program written in MathProg/AMPL;
for comparison with \ref{alg:splicing}, the Stacker Crane Problem is solved exactly using the GLPK software on an integer linear program.
(Simulations were run on a laptop computer with a 2.66 GHz dual core processor and 2 GB of RAM.)

Figure~\ref{fig:splicecostfactor} shows the ratios $\tourlen_\text{SPLICE} / \tourlen^*$ observed in a variety of randomly generated samples (twenty-five trials in each size category).
One can see that the ratio is consistently below $20\%$ even for small problem instances ($n\simeq 10$) and is reduced to  $\simeq 5\%$ for $n>80$.
Hence, convergence to the optimal solutions with respect to the problem size is fairly rapid.
In practice, one could combine \ref{alg:splicing} with an exact algorithm, and let the exact algorithm compute the solution if $n$ is less than, say, $50$, and let \ref{alg:splicing} compute the solution when $n\geq 50$.
\begin{figure}
\centering
\subfigure[Cost factor for SPLICE as a function of the problem size $n$.
Each column records $25$ random observations.]{
\includegraphics[width=.45\linewidth]{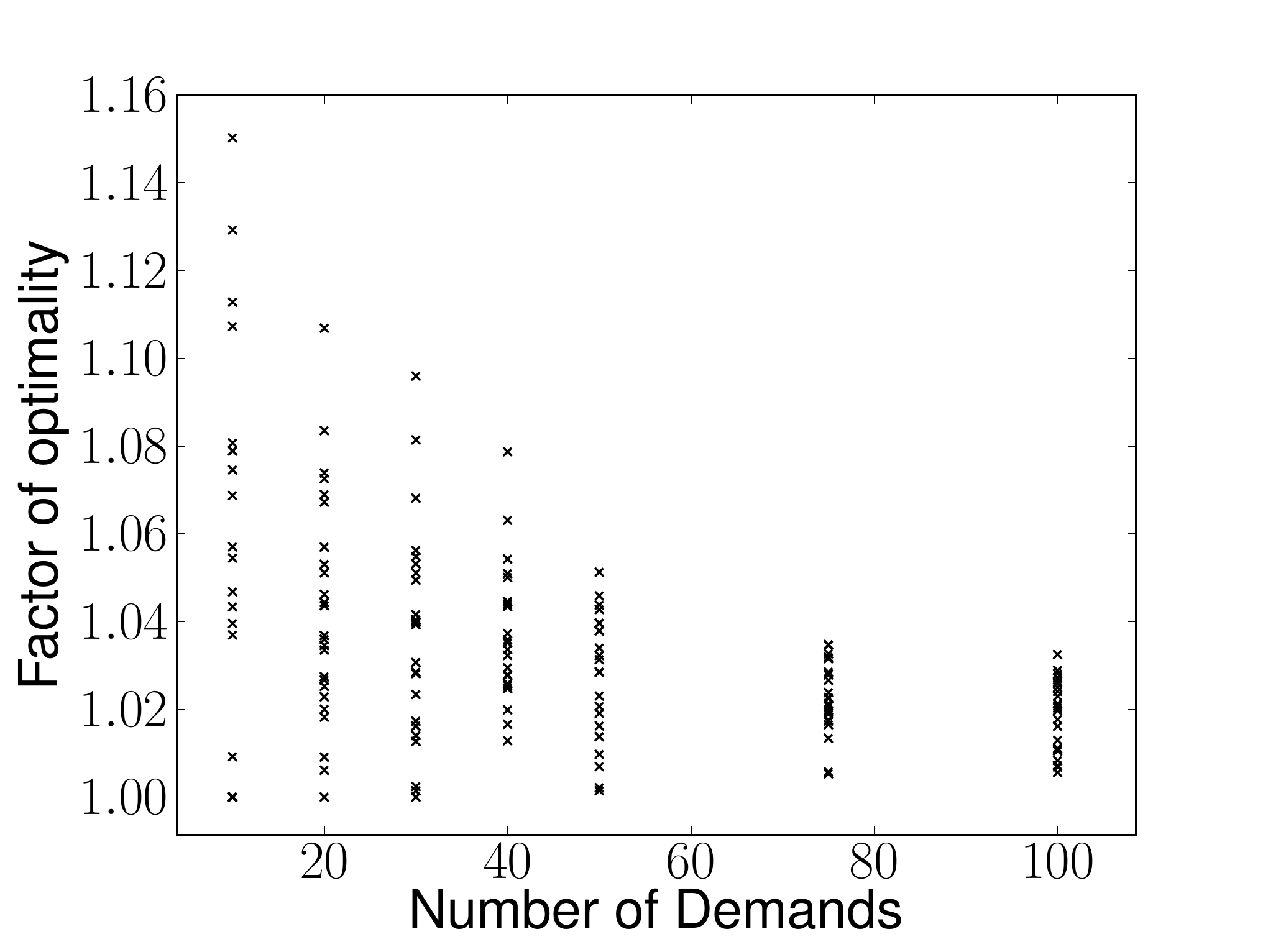}
\label{fig:splicecostfactor}
}
\subfigure[Runtime factor for SPLICE as a function of the problem size $n$.
Each column records $25$ random observations.]{
\includegraphics[width=.45\linewidth]{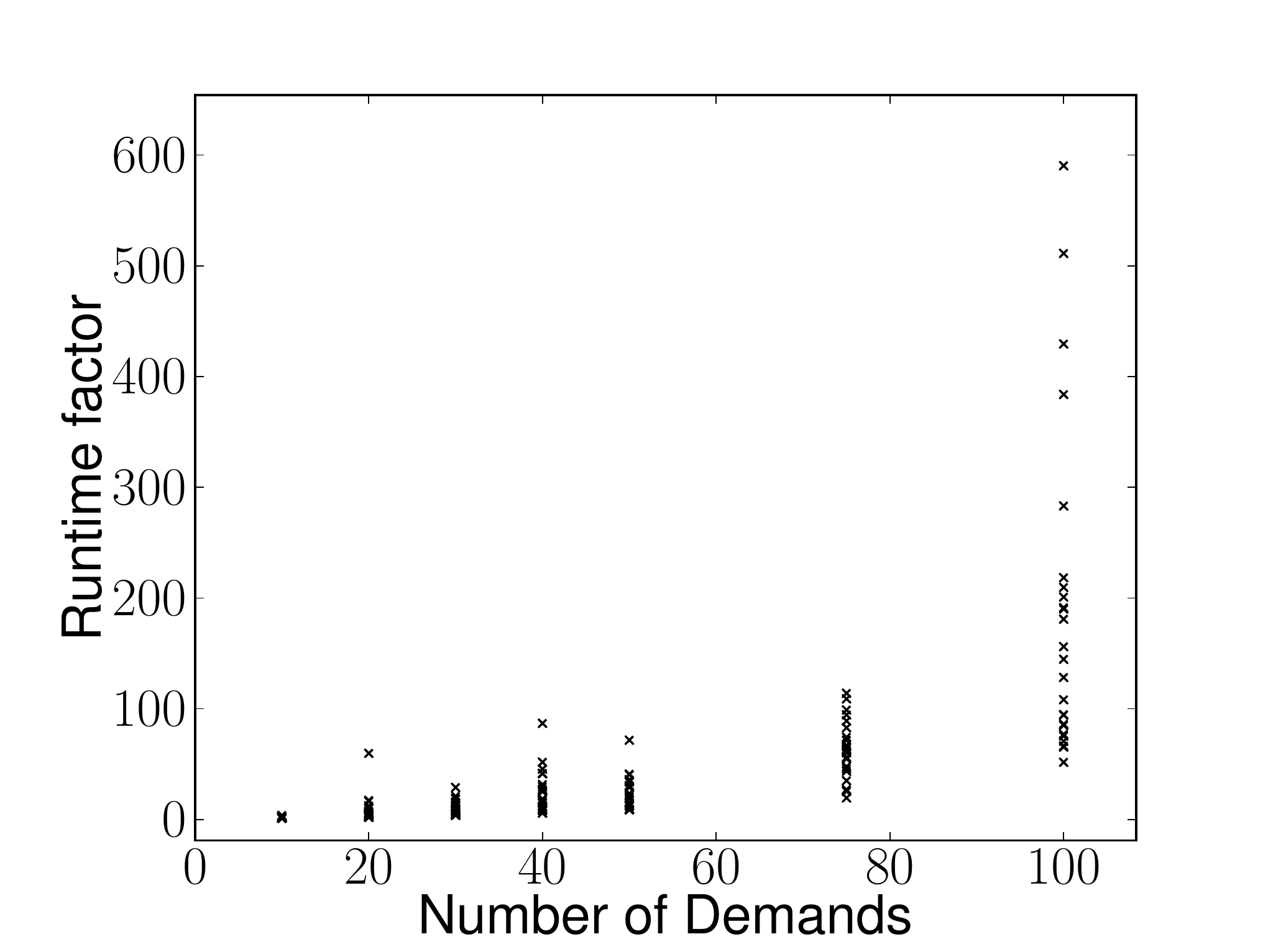}
\label{fig:splicerunfactor}
}
\subfigure[Runtime of the SPLICE algorithm as a function of the problem size $n$.
Each column records $25$ random observations.]{
\includegraphics[width=.45\linewidth]{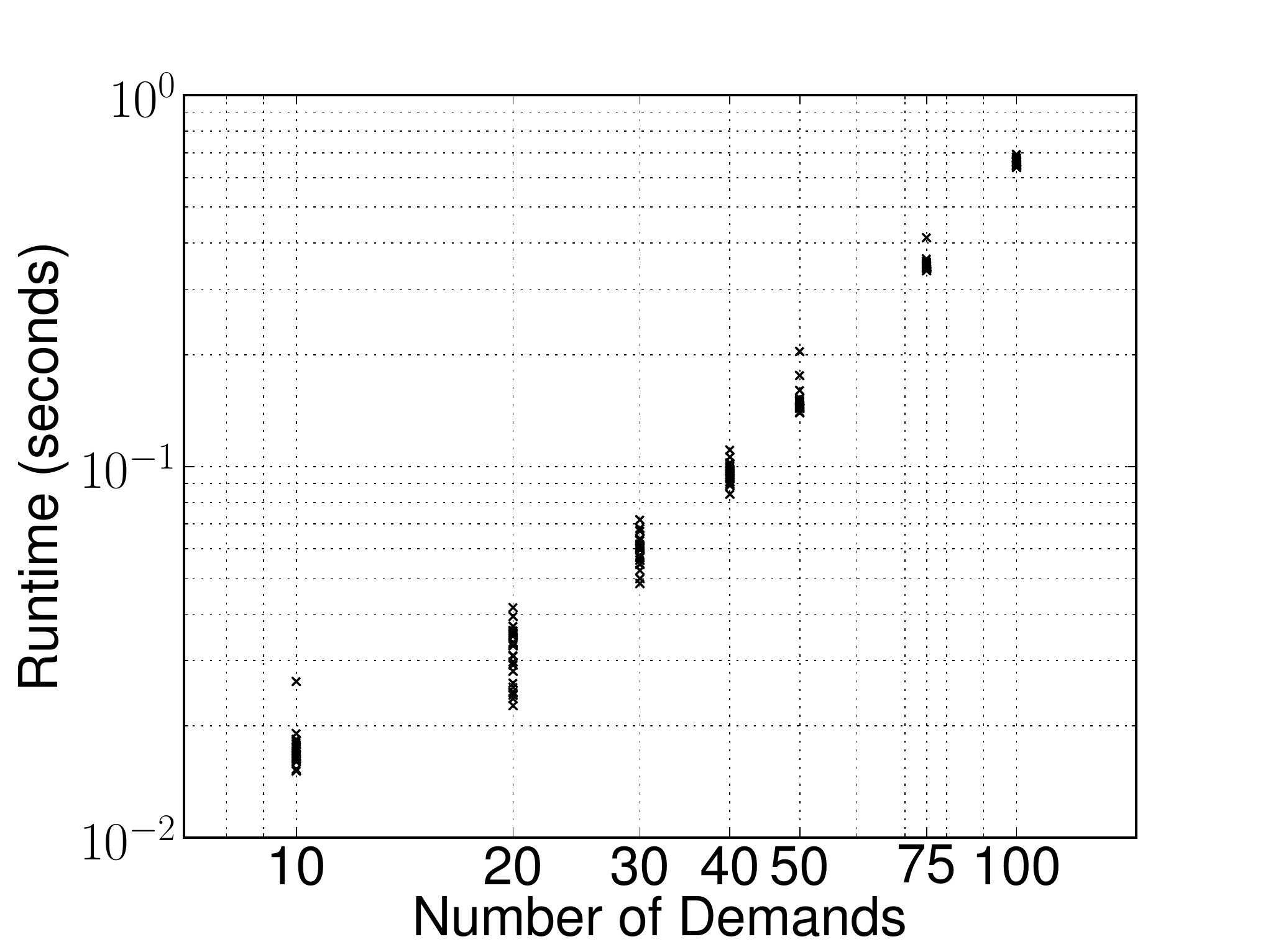}
\label{fig:spliceruntime}
}
\caption{Performance of the SPLICE algorithm over twenty-five trials in each size category.
Figure~\ref{fig:splicecostfactor} shows observed ratios $\tourlen_\text{SPLICE} / \tourlen^*$;
Figure~\ref{fig:splicerunfactor} shows the ratios $T_\text{SPLICE} / T^*$;
Figure~\ref{fig:spliceruntime} shows the runtime $T_\text{SPLICE}$ of the \ref{alg:splicing} algorithm.
}
\end{figure}

Figure~\ref{fig:splicerunfactor} shows the ratios $T_\text{SPLICE} / T^*$ with respect to the size of the problem instance $n$ (the problem instances are the same as those in Figure~\ref{fig:splicecostfactor}), where $T^*$  is the runtime of an exact algorithm.
One can observe that the computation of an optimal solution becomes impractical for a number $n \simeq 100$ of origin/destination pairs.

Finally, Figure~\ref{fig:spliceruntime} shows the runtime $T_\text{SPLICE}$ of the \ref{alg:splicing} algorithm with respect to the size of the problem instance $n$ (the problem instances are the same as those in figure \ref{fig:splicecostfactor}). One can note that even for moderately large problem instances (say, $n\simeq100$) the runtime is below a second.

\subsection{Euclidean Bipartite Matching---First- and Next- Order Asymptotics}
\label{sec:sim_EBMP_scaling}

\newcommand{\simpickrad}{r}
\newcommand{\simdelvrad}{R}
\newcommand{\spherevol}{\operatorname*{spherevol}}

In this section, we compare the observed scaling of the length of the EBMP as a function of instance size,
with what is predicted by equations~\eqref{eq:scpLB} and~\eqref{eq:spliceUB} of Sections~\ref{subsec:linklower} and~\ref{subsec:linkupper}, respectively.
We focus our attention on two examples of pickup/delivery distributions $(\pdistr_\picktag,\pdistr_\delvtag)$:

\noindent\emph{Case I---Unit Cube Arrangement:}
In the first case, the pickup site distribution $\pdistr_\picktag$ places one-half of its probability uniformly over a unit cube centered along the $x$-axis at $x=-4$,
and the other half uniformly over the unit cube centered at $x=-2$.
The delivery site distribution $\pdistr_\delvtag$ places one-half of its probability uniformly over the cube at $x=-4$ and the other half over a new unit cube centered at $x=2$.

\noindent\emph{Case II---Co-centric Sphere Arrangement:}
In the second case, pickup sites are uniformly distributed over a sphere of radius $R=2$, and delivery sites are uniformly distributed over a sphere of radius $r=1$.
Both spheres are centered at the origin.
\begin{figure}
\centering
\subfigure[\emph{Case I} sample ($n=100$) with and without optimal matching.]{
\includegraphics[width=\linewidth]{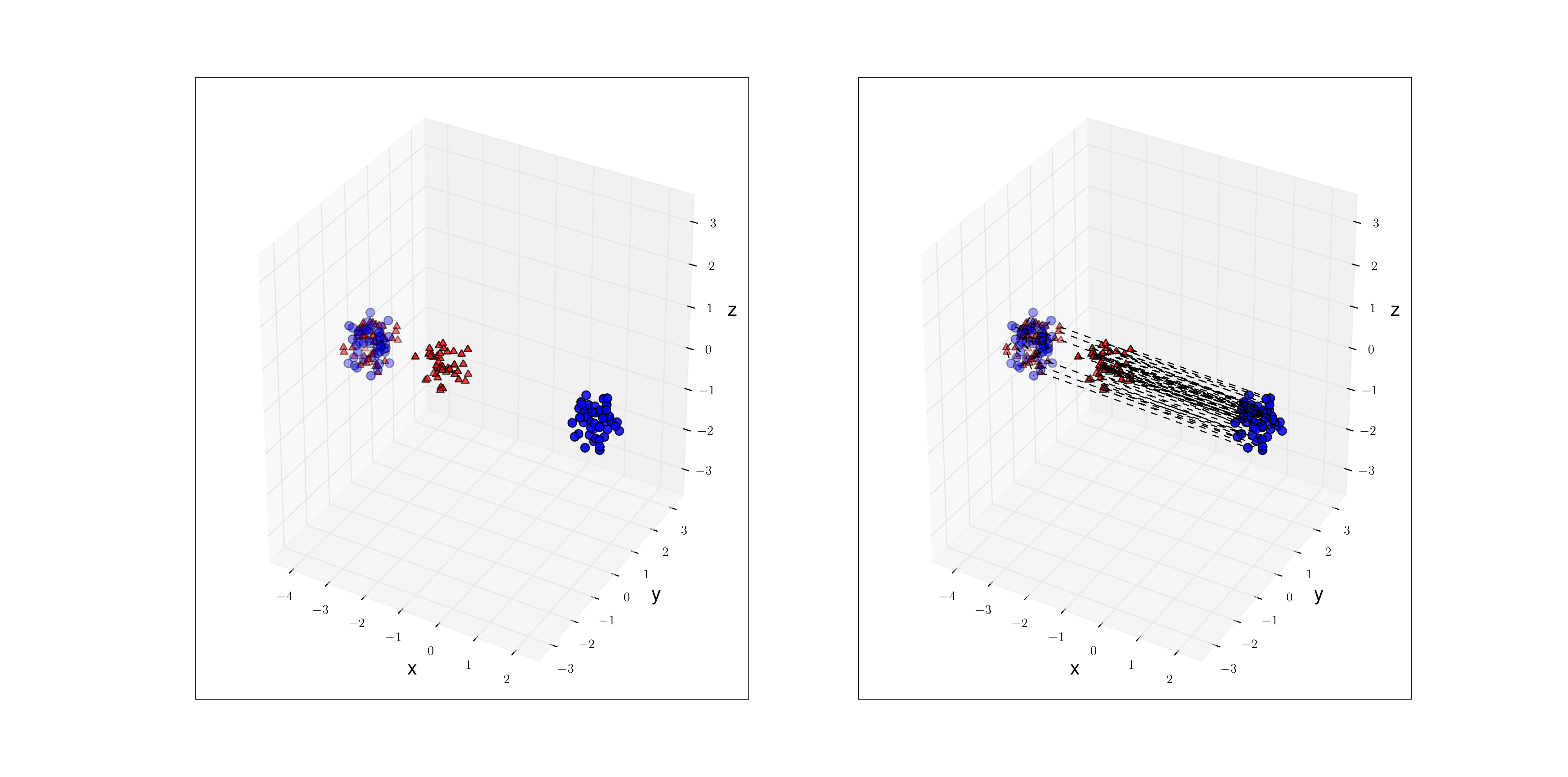}
\label{subfig:cubesmatching}
}
\subfigure[\emph{Case II} sample ($n=100$) with and without optimal matching.]{
\includegraphics[width=\linewidth]{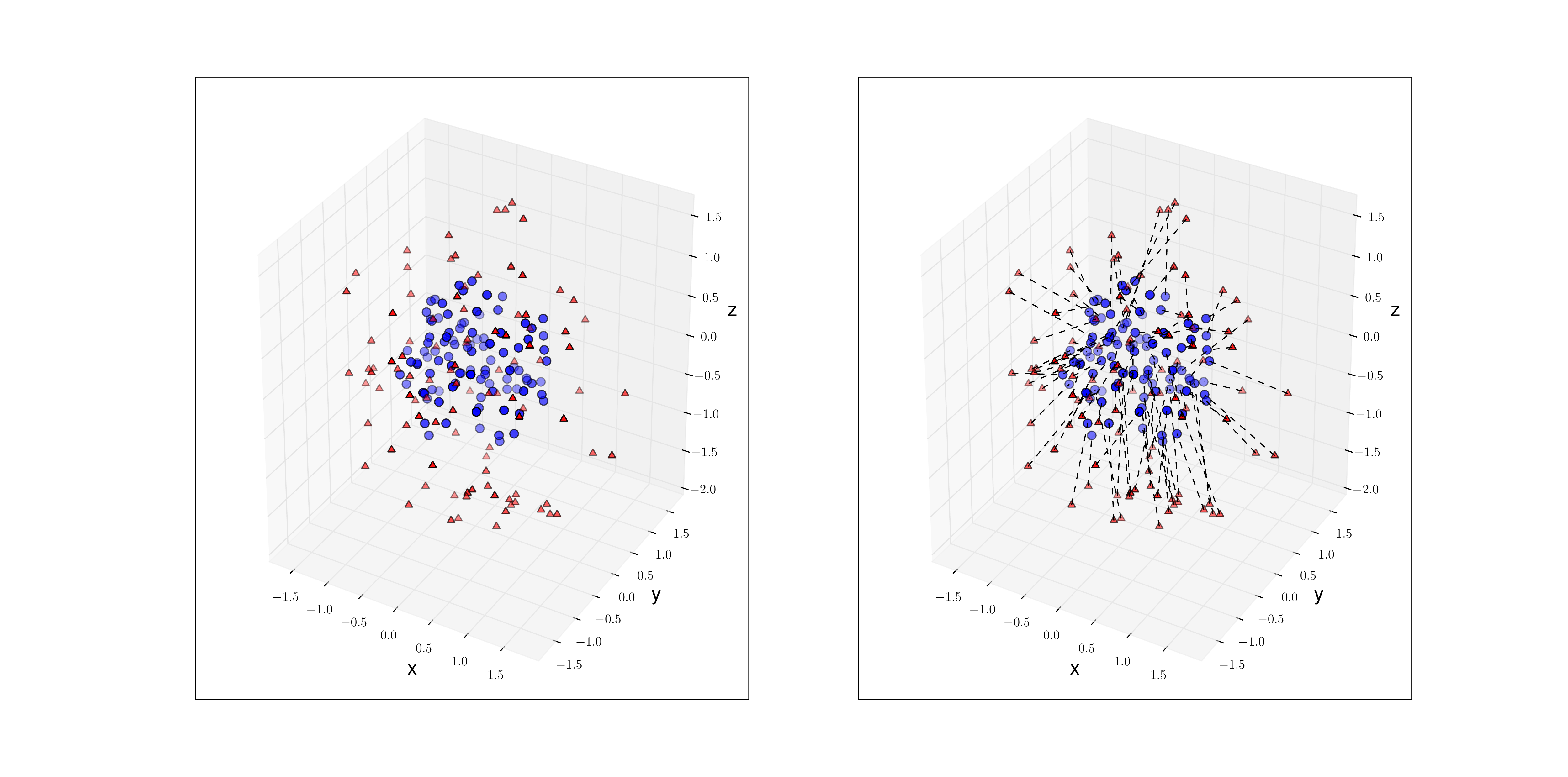}
\label{subfig:spheresmatching}
}
\caption{
Samples of size $n=100$, drawn according to the distributions of \emph{Case I} (Figure~\ref{subfig:cubesmatching}) and \emph{Case II} (Figure~\ref{subfig:spheresmatching}).
Pickup sites are shown as triangle markers; delivery sites are shown as circles.
Left plots show the samples alone; right plots include links of the optimal bipartite matching.
}
\end{figure}

Figures~\ref{subfig:cubesmatching} and~\ref{subfig:spheresmatching}
show samples of size $n=100$, drawn according to the distributions of \emph{Case I} and \emph{Case II} respectively;
the left plots show the samples alone, while the plots on the right include the links of the optimal bipartite matching.

The cases under consideration are examples for which one can compute the constants $W$ (Wasserstein distance) and $\kappa$ of equation~\eqref{eq:spliceUB} exactly.
In the interest of brevity, we omit the derivations, and simply present the computed values in Table~\ref{tbl:bm_scaling_coefficients}.
The extra column ``$\tilde\kappa(\pdistr_\delvtag,\pdistr_\picktag)$'' of the table shows a new smaller constant that results from bringing the $\min$ operation inside the integral in equation~\eqref{eq:residualconst}.
\begin{table}[H]
\centering
\begin{tabular}{l|c|c|c|c|}
		& $W(\pdistr_\delvtag,\pdistr_\picktag)$		
				& $\kappa(\pdistr_\delvtag,\pdistr_\picktag)$
							& $\tilde\kappa(\pdistr_\delvtag,\pdistr_\picktag)$	\\ \hline
  Case I	& $2$		& $\approx 0.892$	& $\approx 0.446$		\\
  Case II	& $0.75$	& $\approx 1.141$	& $\approx 0.285$
\end{tabular}
\caption{
Values computed for the constants $W(\pdistr_\delvtag,\pdistr_\picktag)$ and $\kappa(\pdistr_\delvtag,\pdistr_\picktag)$ in equations~\eqref{eq:scpLB} and~\eqref{eq:spliceUB}, for \emph{Case I} and \emph{Case II}, respectively;
also $\tilde\kappa(\pdistr_\delvtag,\pdistr_\picktag)$ in each case, the result of bringing the $\min$ operation inside the integral in equation~\eqref{eq:residualconst}.
}
\label{tbl:bm_scaling_coefficients}
\end{table}

The simulation experiment is, for either of the cases above, and for each of seven size categories, to sample twenty-five EBMP instances of size $n$ randomly, and compute the optimal matching cost $\M$ of each.
The results of the experiment are shown in Figure~\ref{fig:linearandresidual}.
\begin{figure}
\centering
\subfigure[]{
\includegraphics[width=.45\linewidth]{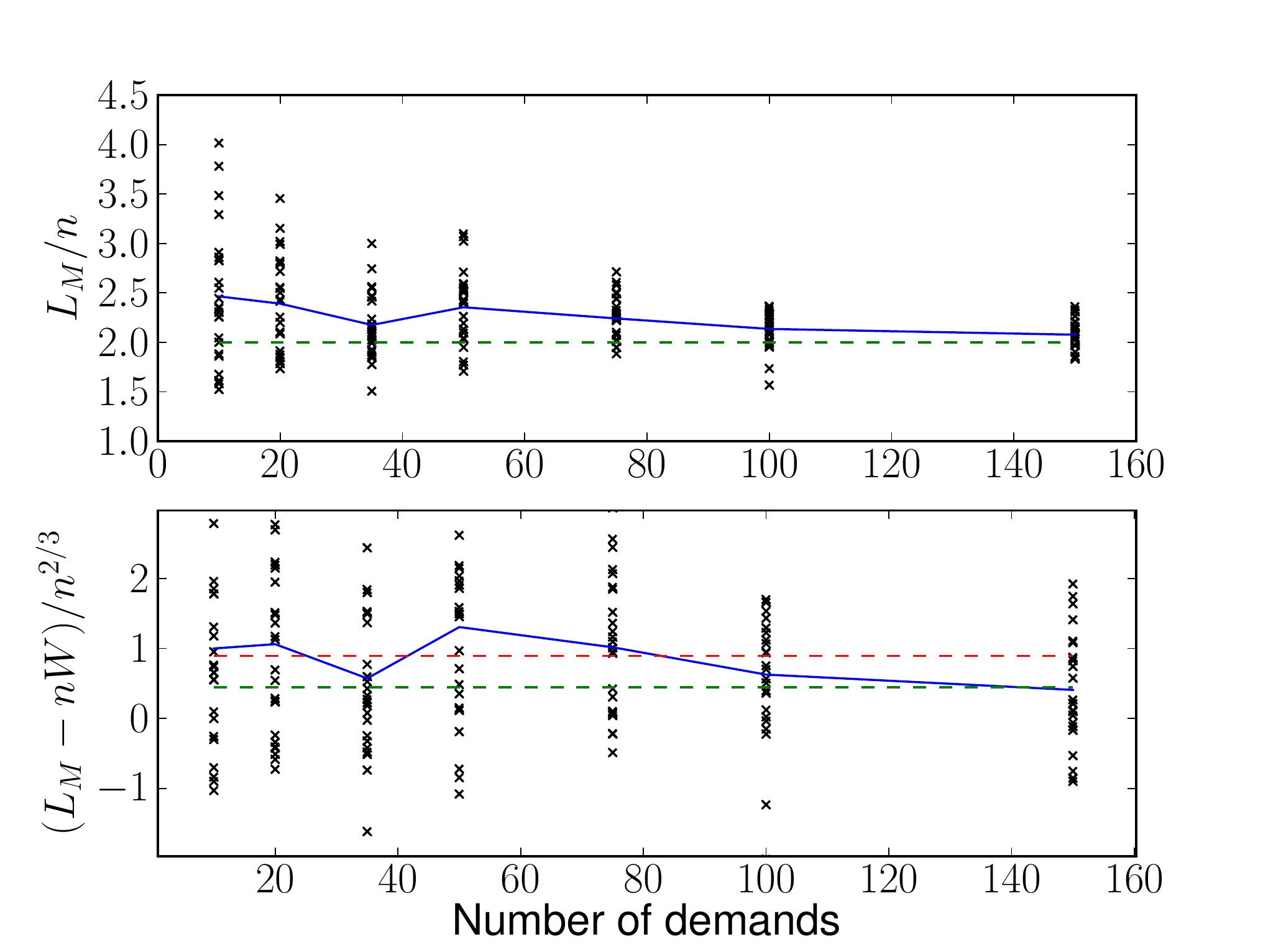}
\label{subfig:experiment_case1}
}
\subfigure[]{
\includegraphics[width=.45\linewidth]{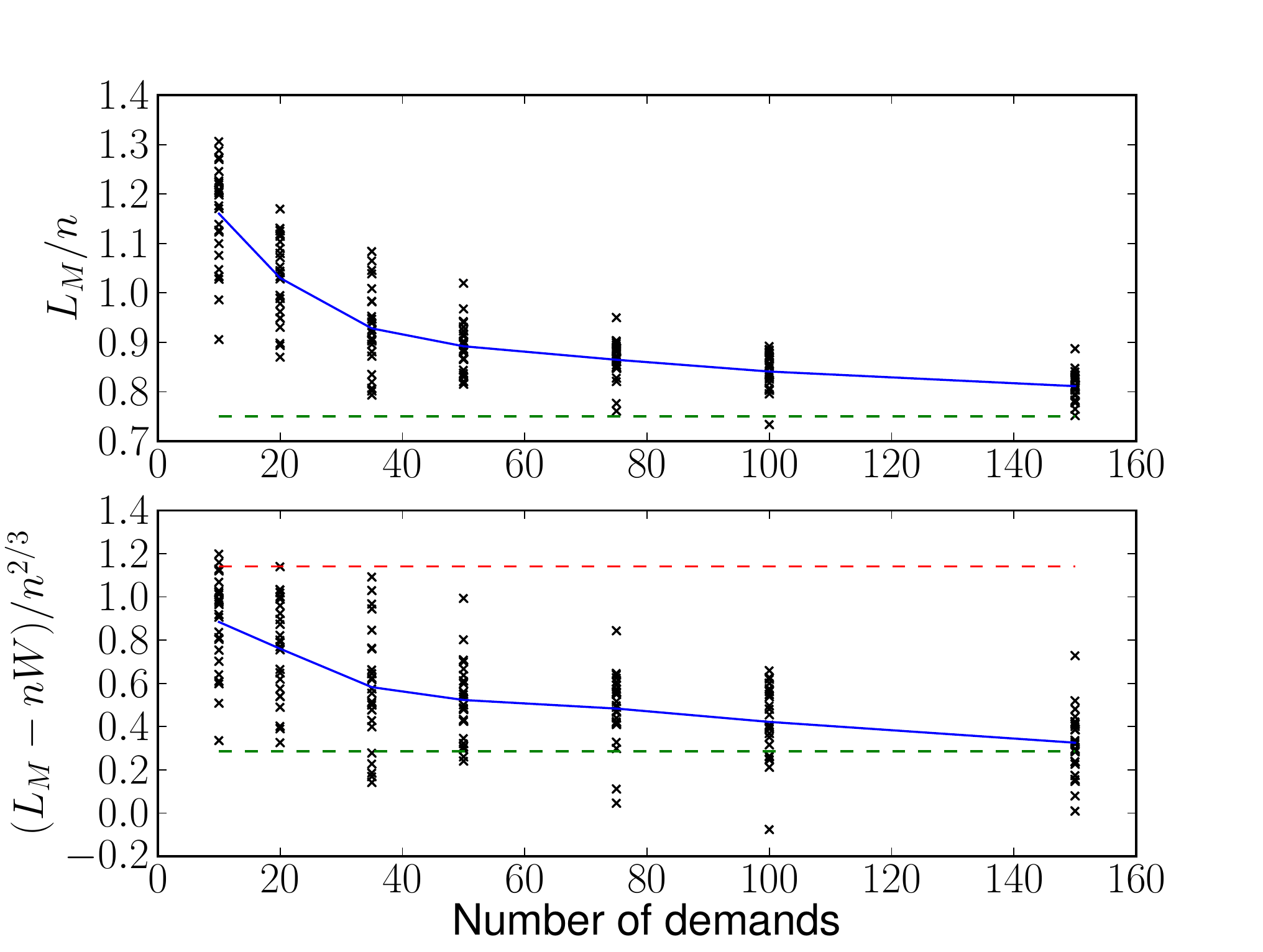}
\label{subfig:experiment_case2}
}
\caption{
Scatter plots of $(n,\M/n)$ (top) and $\left( n, (\M-W)/ n^{2/3} \right)$ (bottom), with one point for each of twenty-five trials per size category.
Figure~\ref{subfig:experiment_case1} shows results for random samples under the distribution of \emph{Case I};
Figure~\ref{subfig:experiment_case2} shows results for random samples under the distribution of \emph{Case II}.
}
\label{fig:linearandresidual}
\end{figure}
Figure~\ref{subfig:experiment_case1} (top) shows a scatter plot of $(n,\M/n)$ with one point for each trial in \emph{Case I};
that is, the $x$-axis denotes the size of the instance, and the $y$-axes denotes the average length of a match in the optimal matching solution.
Additionally, the plot shows a curve (solid line) through the empirical mean in each size category,
and a dashed line showing the Wasserstein distance between $\pdistr_\delvtag$ and $\pdistr_\picktag$, i.e. the predicted asymptotic limit to which the sequence should converge.
Figure~\ref{subfig:experiment_case2} (top) is analogous to Figure~\ref{subfig:experiment_case1} (top), but for random samples of \emph{Case II}.
Both plots exhibit the predicted approach of $\M / n$ to the constant $W(\pdistr_\delvtag,\pdistr_\picktag)>0$; the convergence in Figure~\ref{subfig:experiment_case2} (top) appears slower because $W$ is smaller.
Figure~\ref{subfig:experiment_case1} (bottom) shows a scatter plot of $\left( n, (\M-W)/ n^{2/3} \right)$ from the same data,
with another solid curve through the empirical mean.
Also shown are $\kappa$ and $\tilde\kappa$ (dashed lines);
recall that $\kappa$ is the theoretical asymptotic upper bound for the sequence (equation~\eqref{eq:spliceUB}).
Figures~\ref{subfig:experiment_case2} (bottom) is again analogous to Figures~\ref{subfig:experiment_case1} (bottom),
and both plots indicate asymptotic convergence to a constant no larger than $\kappa(\pdistr_\delvtag,\pdistr_\picktag)$.
In fact, these cases give some credit to a developing conjecture of the authors.
The conjecture is that the minimization in~\eqref{eq:residualconst} can be moved inside the integral to provide an upper bound like equation~\eqref{eq:spliceUB}, but with a smaller (often much smaller) constant factor, i.e. $\tilde \kappa(\pdistr_\delvtag,\pdistr_\picktag)$.
%
%

\subsection{Stability of the DPDP}

We conclude the simulation section with some heuristic validation of equation~\eqref{eq:rho} and the resulting threshold $\arrivalrate^*$ separating \emph{stabilizable} arrival rates from \emph{unstabilizable} ones.
The main insight of this section is as follows.
Let $\pi$ be a policy for the DPDP that is perfectly stabilizing, i.e., stabilizing for all $\arrivalrate < \arrivalrate^*$.
We consider the system $( {\rm DPDP}(\arrivalrate), \pi )$, where $\arrivalrate > \arrivalrate^*$.
Clearly, since $\arrivalrate > \arrivalrate^*$, the number of outstanding demands in the system grows unbounded.
Still, demands arrive at rate $\arrivalrate$ in time average, and we should expect the policy to serve demands at an average rate of $\arrivalrate^*$ (i.e., the fastest rate under $\pi$).
Thus, the number of outstanding demands should grow at an average rate of $\arrivalrate - \arrivalrate^*$.
Since we can control $\arrivalrate$ in simulation, we can use this insight to estimate $\arrivalrate^*$,
e.g., by the simple calculation $\arrivalrate - n(T) / T$ after sufficiently large time $T$, where $n(T)$ is the number of outstanding demands at time $T$.
We focus our discussion on the single-vehicle setting, but results for multiple vehicle systems have been equally positive.
Consider again the cases of Section~\ref{sec:sim_EBMP_scaling}.
Table~\ref{tbl:estimating_lambdastar} shows computed threshold $\arrivalrate^*$ for both cases---and the statistics essential in computing it---%
as well as the estimate of $\arrivalrate^*$ after time $T=5000$.
\begin{table}[H]
\centering
\begin{tabular}{l|c|c|c|c|}
		& $\expectation{ Y-X}$		& $W$		& $\arrivalrate^* = \left(\expectation{ Y-X} + W\right)^{-1}$
										& $\arrivalrate^*$ estimate after $T=5000$	\\ \hline
  Case I	& $\approx 3.2$		& $2$		& $\approx 0.190$	& $0.20$			\\
  Case II	& $\approx 1.66$	& $0.75$	& $\approx 0.415$	& $0.42$
\end{tabular}
\caption{
Stabilizability tresholds $\arrivalrate^*$ for \emph{Case I} and \emph{Case II} of Section~\ref{sec:sim_EBMP_scaling}; with relevant statistics.
Also, an estimate of $\arrivalrate^*$ in each case after simulation for time $T=5000$.
}
\label{tbl:estimating_lambdastar}
\end{table}
Our simulations were of the \emph{nearest-neighbor policy} (NN); i.e., the vehicle's $i$th demand is the demand whose pickup location was nearest to the vehicle at the time of delivery of the $(i-1)$th demand.
(The simulated rate of arrivals $\arrivalrate$ was $1$.)
Although a proof that the NN policy is perfectly stabilizing is currently  not available, it has been observed that such policy has good performance for a variety of vehicle routing problems; it also has a fast implementation where large numbers of outstanding demands are concerned.
In both cases, the estimated and computed $\arrivalrate^*$ were quite close (within $5\%$ of each other).

\section{Conclusion} \label{sec:conclusion}

In this paper we have presented the \ref{alg:splicing} algorithm, a polynomial-time, asymptotically optimal algorithm for the stochastic (Euclidean) SCP. We characterized analytically the length of the tour computed by \ref{alg:splicing}, and we used such characterization to determine a necessary and sufficient condition for the existence of stable routing policies for the 1-DPDP, a dynamic version of the stochastic SCP.
Our results would provide a designer of DRT systems with essential information to build business and strategic planning models regarding, e.g., fleet sizing.

This paper leaves numerous important extensions open for further research.
First, we are interested in precisely characterizing the convergence rate to the optimal solution, and in addressing the more general case where the pickup and delivery locations are statistically correlated.
Second, we plan to develop policies for the dynamic version of the SCP whose performance is within a constant factor from the optimal one.
Third, while in the SCP the servicing vehicle is assumed to be omnidirectional (i.e., sharp turns are allowed), we hope to develop approximation algorithms for the SCP where the vehicle has differential motion constraints (e.g., bounded curvature), as is typical, for example, with unmanned aerial vehicles.
In addition to these natural extensions, we hope that the techniques introduced in this paper (e.g., coupling the EBMP with the theory of random permutations) may come to bear in other hard combinatorial problems.


\bibliographystyle{unsrt} 
\bibliography{TAC2011}


\appendices

\section{The Permutation Probability Assignment of Optimal Bipartite Matchings}
\label{appendix:numberofcycles}

In this section we provide the rigorous proof of Lemma~\ref{lemma:equi} in Section~\ref{subsec:asymptotic_optimality},
showing the equiprobability of permutations produced by an optimal bipartite matching algorithm $\algog$.
Let $\pickset_n = \{x_1, \ldots, x_n\}$ and $\delvset_n = \{y_1, \ldots, y_n\}$ be two sets of points in $\env \subset \reals^\probdim$,
e.g. forming an instance of the EBMP in environment $\env$.
Consider $\batchvar = \concat( \xvar_1, \yvar_1, \ldots, \xvar_\batchsize, \yvar_\batchsize )$, a column vector formed by vertical concatenation of $x_1, y_1, \ldots, x_n, y_n$.
Note that the set $\batchsetdesc$ of such vectors, i.e. the span of the instances of the EBMP, is a full-dimensional subset of $\reals^{d(2n)}$.
Let $\matchperm : \batchsetdesc \to 2^{\permset_n }$ denote the \emph{optimal permutation map} that maps a \emph{batch} $s \in \batchsetdesc$ into the \emph{set} of permutations that correspond to optimal bipartite matchings (recall that there might be multiple optimal bipartite matchings).
Let us denote the set of batches that lead to non-unique optimal bipartite matchings as:
\[
\zb := \Bigl \{ s\in \batchsetdesc \, \Big | \, |\matchperm(s)|>1 \Bigr \},
\]
where $|\matchperm(s)|$ is the cardinality of set $\matchperm(s)$.

In Lemma~\ref{lemma:equi}, $\algog$ may be any algorithm that computes an optimal bipartite matching (i.e., a permutation that solves an EBMP).
According to our definitions, the behavior of such an algorithm can be described as follows: given a batch $s \in \batchsetdesc$ it computes 
\begin{equation*}
\algog(s) =  
\begin{cases} 
\text{unique }
\permvar \in \matchperm(s) & \text{if $s\in \batchsetdesc \setminus \zb$,}
\\
\text{some } \sigma\in \matchperm(s) &\text{otherwise.}
\end{cases}
\end{equation*}
Thus, the behavior of a bipartite algorithm on the set $\zb$ can vary;
on the other hand, we now show that set $\zb$ has Lebesgue measure zero so that the behavior of an algorithm on this set is immaterial for our analysis.
%
\begin{lemma}[Measure of multiple solutions]
\label{lemma:zbmeasure}
The set $\zb$ has Lebesgue measure equal to zero.
\end{lemma}
\begin{proof}
The strategy of the proof is to show that $\zb$ is the subset of a set that has zero Lebesgue measure. 

For $\sigma^{\prime}, \sigma^{\prime \prime} \in \permset_n$, $\sigma^{\prime} \neq \sigma^{\prime \prime}$,
let us define the sets:
\[
\mathcal H_{\sigma^{\prime},  \sigma^{\prime \prime}}
  := \Bigl \{ s \in \batchsetdesc \Big |  \sum_{i=1}^n \| x_{\permvar^{\prime}(i)} - y_i \|  =  \sum_{i=1}^n \| x_{\permvar^{\prime \prime}(i)} - y_i \|  \, \Bigr\};
\]
let us also define the union of such sets:
\[
\mathcal H := \bigcup_{\substack{\sigma^{\prime}, \sigma^{ \prime\prime} \in \permset_n\\ \sigma^{\prime} \neq \sigma^{\prime \prime}}} \, \, \mathcal H_{\sigma^{\prime}, \, \sigma^{\prime \prime}}.
\]

The equality constraint in the definition of $\mathcal H_{\sigma^{\prime}, \, \sigma^{\prime \prime}}$ implies that $\mathcal H_{\sigma^{\prime}, \, \sigma^{\prime \prime}} \subseteq \reals^{d(2n)-1}$, which has zero Lebesgue measure in $\reals^{d(2n)}$.
Hence, the Lebesgue measure of $\mathcal H_{\sigma^{\prime}, \, \sigma^{\prime \prime}}$ is zero.
Since $\mathcal H$ is the union of finitely many sets of measure zero, it has zero Lebesgue measure as well. 

We conclude the proof by showing that $\zb \subseteq \mathcal H$. Indeed, if $s\in \zb$, there must exist two permutations $\sigma^{\prime} \neq \sigma^{\prime \prime}$ such that $ \sum_{i=1}^n \| x_{\permvar^{\prime}(i)} - y_i \| = \min_{\sigma} \sum_{i=1}^n \| x_{\permvar(i)} - y_i \|$ and $ \sum_{i=1}^n \| x_{\permvar^{\prime \prime}(i)} - y_i \| = \min_{\sigma} \sum_{i=1}^n \| x_{\permvar(i)} - y_i \|$, i.e.,  there must exist two permutations $\sigma^{\prime} \neq \sigma^{\prime \prime}$ such that
\[
\sum_{i=1}^n \| x_{\permvar^{\prime}(i)} - y_i \|  =  \sum_{i=1}^n \| x_{\permvar^{\prime \prime}(i)} - y_i \| ,
\]
which implies that $s \in \mathcal H$. Hence, $\zb \subseteq H$ and, therefore, it has zero Lebesgue measure.

\end{proof}

Now we present the proof of Lemma~\ref{lemma:equi}, which gives the probability that $\algog$ produces as a result the permutation $\sigma$ for $\pickset_n,\delvset_n \sim \ESCPinst$;
we call such probability $\prob{\sigma}$.
\begin{proof}[Proof of Lemma~\ref{lemma:equi}]
We start by observing that it is enough to consider a restricted sample space, namely $\batchsetdesc \setminus \zb$.
One can write $\prob{s \in \zb} = \int_{s \in \zb} \pdistr( s) d{ s}$,
where $\pdistr(s)$ denotes the product $\prod_{i=1}^n \pdistr_\picktag(x_i) \pdistr_\delvtag(y_i)$.
Because of our continuity assumptions on probability distributions, Lemma~\ref{lemma:zbmeasure} implies $\prob{s \in \zb} = 0$.
Thus, by the total probability law,
\begin{equation}\label{eq:remZ}
\begin{split}
\prob{\sigma}
&= \probcond{\algog(s) = \sigma}{s \in \batchsetdesc \setminus \zb}.
\end{split}
\end{equation}

For each permutation $\permvar \in   \permset_n$, let us define the set
\[ 
\batchset_{\permvar} := \left\{  \batchvar \ \in \batchsetdesc \setminus \zb \,|\,  \algog(\batchvar) = \permvar  \right\}.
\]
Collectively, sets $\batchset_{\permvar}$ form a partition of $\batchsetdesc \setminus \zb$. This fact, coupled with equation \eqref{eq:remZ}, implies
\[
\prob{\permvar} = \prob{s\in \batchset_{\permvar}}.
\]
For a permutation $\permvar \in \permset_n$, let us define the \emph{reordering function} $g_{\permvar}: \batchsetdesc \to \batchsetdesc$ as the function that maps a batch
$\batchvar = \concat(x_1,y_1,\ldots,x_n, y_n)$
into a batch
$\batchvar' = \concat(x_{\permvar(1)}, y_1,\ldots,x_{\permvar(n)}, y_n)$.
Alternatively, let $E_j \in \reals^{d \times 2nd}$ be a block row matrix of $2n$ $d\times d$ square blocks whose elements are equal to zero except the $(2j-1)$th block that is identity;
let $F_j \in \reals^{d \times 2nd}$ be such a block matrix, but whose elements are all zero except the $2j$th block that is identity.
Then in matrix form the reordering function can be written as
$g_{\permvar}(s) = P_{\permvar} \, s$,
where $P_{\permvar}$ is the $2nd\times 2nd$ matrix defined by
\[
P_{\permvar} :=
\left[ \begin{array}{ccc}
  E_{\permvar(1)}^\xpose	\
  F_1^\xpose			\
  E_{\permvar(2)}^\xpose	\
  F_2^\xpose			\
  \hdots			\
  E_{\permvar(n)}^\xpose	\
  F_n^\xpose	
\end{array} \right]^\xpose.
\] 
Note that $|\det(P_{\permvar})|=1$ for all permutations $\permvar$; also, Prop.~2 of Section~\ref{subsec:cycle} implies $P_{\permvar^{-1}} = P_\permvar^{-1}$.
%
We now show that $g_{\hat \permvar}(\batchset_{\hat \permvar}) = \batchset_{\permvar_1}$ for all permutations $\hat \permvar \in \permset_n$,
recalling that $\permvar_1$ denotes the \emph{identity} permutation (i.e., $\permvar(i) = i$ for $i=1,\ldots,\batchsize$):
%
\begin{itemize}
\item
$g_{\hat \permvar}(\batchset_{\hat \permvar}) \subseteq \batchset_{\permvar_1}$.
Let $\hat s \in \batchset_{\hat \permvar}$.
Then by definition
$\sum_{i=1}^n \vecnorm{ \hat x_{\hat \permvar(i)} - \hat y_i }
  = \min_{ \permvar \in \permset_n } \sum_{i=1}^n \vecnorm{ \hat x_{\permvar(i)} - \hat y_i }$; moreover $\hat \permvar$ is the unique minimizer.
We want to show that $g_{\hat \permvar}(\hat s) \in \batchset_{\permvar_1}$,
where $g_{\hat \permvar}(\hat s)$ has the form	\\
$\concat(\hat x_{\hat \permvar(1)}, \hat y_1,\ldots,\hat x_{\hat \permvar(n)}, \hat y_n)$.
Let $s = g_{\hat \permvar}(\hat s)$;
indeed, $\permvar_1$ is an optimal matching of $s$ (by inspection), i.e., $\permvar_1 \in \permset^*(s)$.
Suppose, however, there is another optimal matching ${\hat{\hat \permvar}} \neq \permvar_1$ such that ${\hat{\hat \permvar}} \in \permset^*(s)$.
Then ${\hat{\hat \permvar}} \hat \permvar$ is an optimal matching of $\hat s$ (Prop.~1);
yet this is a contradiction,
because ${\hat{\hat \permvar}} \hat\permvar \neq \hat\permvar$.
Therefore, we have that $s \in \batchset_{\permvar_1}$ for all $\hat s \in \batchset_{\hat \permvar}$.
%
\item
$\batchset_{\permvar_1} \subseteq g_{\hat \permvar}(\batchset_{\hat \permvar})$.
Let $s \in \batchset_{\permvar_1}$.
Then by definition
$\sum_{i=1}^n \vecnorm{ x_i - y_i }
  = \min_{ \permvar \in \permset_n } \sum_{i=1}^n
      \vecnorm{ x_{\permvar(i)} - y_i }$;
moreover $\permvar_1$ is the unique minimizer.
Note that $g_{\hat \permvar}$ is an injective function
(since the determinant of $P_{\hat \permvar}$ is nonzero);
let $\hat s$ be the unique batch such that $s = g_{\hat \permvar}(\hat s)$, i.e.,
$\hat s = \concat(x_{\hat \permvar^{-1}(1)}, y_1,\ldots,x_{\hat \permvar^{-1}(n)},y_n)$ (Prop.~2).
We want to show that $\hat s \in \batchset_{\hat \permvar}$.
%
Because $\sum_{i=1}^n \vecnorm{ x_i - y_i }
  = \sum_{i=1}^n \vecnorm{ x_{\hat \permvar( \hat \permvar^{-1}(i) ) } - y_i }$,
$\hat \permvar$ is an optimal matching of $\hat s$, i.e., $\hat \permvar \in \permset^*(\hat s)$.
Suppose there is another optimal matching
${\hat{\hat \permvar}} \neq \hat \permvar$
such that ${\hat{\hat \permvar}} \in \permset^*(\hat s)$.
Again, this is a contradiction, since $\hat{\hat \permvar} \hat \permvar^{-1} \neq \permvar_1$, and $\permvar_1$ is the \emph{unique} optimal matching for batch $s$. We conclude that $\hat s \in \batchset_{\hat \permvar}$ for all $s \in \batchset_{\permvar_1}$.
\end{itemize}

We are ready to evaluate the probabilities of permutations as follows:
For any permutation $\hat \permvar$ we have
$\prob{\hat \permvar}
  = \prob{\hat s \in \batchset_{\hat \permvar}}
  = \int_{\hat s \in \batchset_{\hat \permvar}} \pdistr(\hat s) d{\hat s}$,
where $\pdistr(\hat s)$ denotes 
$\prod_{i=1}^n \pdistr_\picktag(\hat x_i) \pdistr_\delvtag(\hat y_i)$.
We use variable substitution
$s = g_{\hat \permvar}(\hat s) = P_{\hat \permvar} \hat s$
and the property
$g_{\hat \permvar}(\batchset_{\hat \permvar}) = \batchset_{\permvar_1}$,
and we apply the rule of integration by substitution:
$\int_{\hat s \in \batchset_{\hat \permvar}} \pdistr(\hat s) d{\hat s}
  = \int_{s \in \batchset_{\permvar_1}} \pdistr( P_{\hat \permvar}^{-1} s )
    \underbrace{|\det(P_{\hat \permvar})|^{-1}}_{=1} \ ds$.
Observing that
\[
\pdistr( P_{\hat \permvar}^{-1} s )
  = \pdistr( P_{{\hat \permvar}^{-1}} s )
  = \prod_{i=1}^n
      \pdistr_\picktag( x_{ {\hat \permvar}^{-1}(i) } )
      \pdistr_\delvtag( y_i ),
\]
and that
\[\prod_{i=1}^n
    \pdistr_\picktag( x_{ {\hat \permvar}^{-1}(i) } )
    \pdistr_\delvtag( y_i )
  = \prod_{i=1}^n \pdistr_\picktag( x_i ) \pdistr_\delvtag( y_i )
  = \pdistr(s),\]
we obtain
\begin{equation*}
\begin{split}
\int_{s \in \batchset_{\permvar_1}}
    \pdistr( P_{\hat \permvar}^{-1} s ) \ ds
  &= \int_{s \in \batchset_{\permvar_1}} \pdistr( s ) \ ds
  = \prob{ s \in \batchset_{\permvar_1} } = \prob{\permvar_1}.
    \end{split}
 \end{equation*}
Combining these results,
we conclude $\prob{\sigma} = \prob{\sigma_1}$ for all $\permvar \in \permset_n$,
obtaining the lemma.
\end{proof}


\section{Proofs of Other Lemmas}
\label{appendix:lemmas}

\subsection{Lemmas of Section~\ref{sec:scpoptimal}}

\begin{proof}[Proof of Lemma~\ref{lemma:splicesvanish}]
For any $\epsilon>0$, consider the sequence $E$ of events, where
\begin{align*}
    E_\batchsize \doteq \quad\Bigr\{ (\mathcal X_n, \mathcal Y_n): \, 
		\subn_\batchsize/
		\batchsize
	 > \epsilon \Bigl \}
\end{align*}
or, equivalently,
$E_\batchsize
  = \Bigr\{ (\mathcal X_n, \mathcal Y_n) : \,
      \left( \subn_\batchsize  - \expect \subn_\batchsize \right)
	+ \left( \expect \subn_\batchsize - \log(\batchsize) \right)
	+ \log(\batchsize)
      > \epsilon \batchsize
  \Bigl\}$.
%
By Lemma~\ref{lemma:sub2perm}, the number of disconnected subtours is \emph{equal} to the number of cycles in the permutation $\permvar$ computed by the matching algorithm $\algog$ in line \ref{line:perm}. Since, by Lemma \ref{lemma:equi}, all permutations are equiprobable, the number of cycles has expectation and variance both equal to $\log(n) + O(1)$. Therefore, we conclude that $\subn_n $ has expectation and variance both $\log(n) + O(1)$.
Hence, we can rewrite the events $E_n$ as:
\[
E_\batchsize = \Bigr\{(\mathcal X_n, \mathcal Y_n):\, \subn_\batchsize  - \expect  \subn_\batchsize 
	> \epsilon\, \batchsize+ o(\batchsize) \Bigl \}.
\]
Applying Chebyshev's inequality, we obtain  (for $\batchsize'$ sufficiently large, yet finite)
\begin{align*}
    \sum_{\batchsize=0}^\infty \probability\left[ E_\batchsize \right]
    &\leq
    n' + \sum_{\batchsize=n'}^\infty
       \frac{ \log(\batchsize) + O(1) }{ \left[ \epsilon\, \batchsize + o(\batchsize) \right]^2 }.
\end{align*}
%
Since this summation is finite,
we can apply the Borel-Cantelli lemma to the sequence of events $E_n$ and conclude that $ \probability[\limsup_{n\to +\infty} E_n]=0$.
Finally, since $\epsilon$ can be chosen arbitrarily small, the upper limit of the claim follows (the lower limit holds trivially). 
\end{proof}

\subsection{Lemmas of Section~\ref{sec:scpbounds}}

\begin{proof}[Proof of Lemma~\ref{lemma:avgmatch_weaklowerbound}]
Let $\permvar$ denote the optimal bipartite matching of $Q_\numdem$.
For a particular partition $\cellset$, we define random variables
$\hat\alpha_{\ivar\jvar} := \left| \left\{ \kvar : \Yvar_\kvar \in \workcell^\ivar, \Xvar_{\permelem{\kvar}} \in \workcell^\jvar \right\} \right| / \numdem$
for every pair $(\workcell^\ivar,\workcell^\jvar)$ of cells;
that is, $\hat\alpha_{\ivar\jvar}$ denotes the fraction of matches under $\permvar$ whose $\delvset$-endpoints are in $\workcell^\ivar$ and whose $\pickset$-endpoints are in $\workcell^\jvar$.
Let $\hat\transportset_\numdem$ be the set of matrices with entries $\left\{ \alpha_{\ivar\jvar} \geq 0 \right\}_{\ivar,\jvar = 1,\ldots,|\cellset|}$,
such that $\sum_\ivar \alpha_{\ivar\jvar} = \left| \pickset_\numdem \cap \workcell^\jvar \right| / \numdem$ for all $\workcell^\jvar \in \cellset$
and $\sum_\jvar \alpha_{\ivar\jvar} = \left| \delvset_\numdem \cap \workcell^\ivar \right| / \numdem$ for all $\workcell^\ivar \in \cellset$;
note $\left\{ \hat\alpha_{\ivar\jvar} \right\}$ itself is an element of $\hat\transportset_\numdem$.
Then the average match length $\avglink_\text{M}(Q_\numdem)$ is bounded below by
\begin{equation*}
\begin{aligned}
\avglink_\text{M}(Q_\numdem) =
\oneover{\numdem} \sum_{\kvar=1}^\numdem \vecnorm{ \Xvar_{\sigma(\kvar)} - \Yvar_\kvar }
  & \geq \sum_{\ivar\jvar} \hat\alpha_{\ivar\jvar} \min_{y\in\workcell^\ivar, x\in\workcell^\jvar} \vecnorm{x-y}	\\
  & \geq \min_{A \in \hat\transportset_\numdem} \sum_{\ivar\jvar} \alpha_{\ivar\jvar} \min_{y\in\workcell^\ivar, x\in\workcell^\jvar} \vecnorm{x-y}.
\end{aligned}
\end{equation*}
The key observation is that
$\lim_{\numdem\to\infty} \left| \left\{ \pickset_\numdem \cap \workcell^\jvar \right\} \right| / \numdem = \pdistr_\picktag( \workcell^\jvar )$,
and $\lim_{\numdem\to\infty} \left| \left\{ \delvset_\numdem \cap \workcell^\ivar \right\} \right| / \numdem = \pdistr_\delvtag( \workcell^\ivar )$, almost surely.
Applying standard sensitivity analysis (see Chapter~5 of~\cite{bertsimas1997introduction}), we observe that the final expression converges almost surely to $\underline\avglink(\cellset)$
as $n \to +\infty$.

\newcommand{\picksense}{ {\lambda_\picktag} }
\newcommand{\pickpertb}{ {\Delta_\picktag} }
\newcommand{\pickbound}{ {\epsilon_\picktag} }
\newcommand{\delvsense}{ {\lambda_\delvtag} }
\newcommand{\delvpertb}{ {\Delta_\delvtag} }
\newcommand{\delvbound}{ {\epsilon_\delvtag} }
Specifically, for any finite partition $\cellset$ and $\epsilon > 0$, consider the sequence of events $E$, where
\[
  E_n =
  \left\{
    (\pickset_\numdem,\delvset_\numdem) \ : \
    \left|
      \min_{A \in \hat\transportset_\numdem} \sum_{\ivar\jvar} \alpha_{\ivar\jvar} \min_{y\in\workcell^\ivar, x\in\workcell^\jvar} \vecnorm{x-y}
      - \underline\avglink(\cellset)
    \right| > \epsilon
  \right\}.
\]
Consider Problem~\ref{problem:traffic_optimistic}; 
let $\delvsense_i$ for all $i$ be the dual variables associated with the constraints $\sum_j \alpha_{ij} = \pdistr_\delvtag(\workcell^i)$;
let $\picksense_j$ for all $j$ be the dual variables associated with the constraints $\sum_i \alpha_{ij} = \pdistr_\picktag(\workcell^j)$.
(These are all finite constants.)
Also, let $\delvpertb_i(n) = \left| \delvset_n \cap \workcell^i \right|/n - \pdistr_\delvtag(\workcell^i)$ for all $i$;
let $\pickpertb_j(n) = \left| \pickset_n \cap \workcell^j \right|/n - \pdistr_\picktag(\workcell^j)$ for all $j$.
Through sensitivity analysis we obtain
\begin{equation*}
\begin{aligned}
\min_{A \in \hat\transportset_\numdem} \sum_{\ivar\jvar} \alpha_{\ivar\jvar} \min_{y\in\workcell^\ivar, x\in\workcell^\jvar} \vecnorm{x-y}
      - \underline\avglink(\cellset)
= \sum_i \delvsense_i \delvpertb_i(n) + o(\delvpertb_i(n)) + \sum_j \picksense_j \pickpertb_j(n) + o(\pickpertb_j(n))
\end{aligned}
\end{equation*}
and so
\begin{equation}
\label{eq:sensitivity_bound}
\begin{aligned}
\left|
  \min_{A \in \hat\transportset_n} \sum_{\ivar\jvar} \alpha_{\ivar\jvar} \min_{y\in\workcell^\ivar, x\in\workcell^\jvar} \vecnorm{x-y}
      - \underline\avglink(\cellset)
\right|
&\leq
  \sum_i | \delvsense_i | | \delvpertb_i(n) | + o(\delvpertb_i(n))		\\
& \qquad\qquad {} + \sum_j | \picksense_j | | \pickpertb_j(n) | + o(\pickpertb_j(n)).
\end{aligned}
\end{equation}
Thus we can choose $\delvbound_i > 0$ for all $i$ and $\pickbound_j > 0$ for all $j$ so that the right-hand side of~\eqref{eq:sensitivity_bound} is less than or equal to $\epsilon$
as long as $| \delvpertb_i | \leq \delvbound_i$ for all $i$ and $| \pickpertb_j | \leq \pickbound_j$ for all $j$.
Let us now define the events $E_{\delvpertb_i}$ for all $i$ and $E_{\pickpertb_j}$ for all $j$, where
\[
\begin{aligned}
  E_{\delvpertb_i n} =
  \left\{
    (\pickset_\numdem,\delvset_\numdem) \ : \
    | \delvpertb_i | > \delvbound_i
  \right\}
&& \text{and} &&
  E_{\pickpertb_j n} =
  \left\{
    (\pickset_\numdem,\delvset_\numdem) \ : \
    | \pickpertb_j | > \pickbound_j
  \right\}.
\end{aligned}
\]
Note that the Strong Law of Large Numbers gives
$\prob{ \limsup_{n \to +\infty} E_{\delvpertb_i n} } = 0$ for all $i$, and
$\prob{ \limsup_{n \to +\infty} E_{\pickpertb_j n} } = 0$ for all $j$.
Observing that
\[
  \prob{ \limsup_{n \to +\infty} E_n } \leq
    \sum_i \prob{ \limsup_{n \to +\infty} E_{\delvpertb_i n} } 
    + \sum_j \prob{ \limsup_{n \to +\infty} E_{\pickpertb_j n} } = 0
\]
we obtain the claim.
\end{proof}

\begin{proof}[Proof of Lemma~\ref{lemma:avgmatch_lowerbound}]
First, we show that $\underline\avglink \leq W(\pdistr_\delvtag,\pdistr_\picktag)$.
Fix $\gamma^*$, some solution arbitrarily close to the infimum of~\eqref{eq:wasserstein}.
Given a partition $\cellset$ let us define the distance function
$\tilde d(x_1,x_2)
  = \min_{x \in \cellset(x_1), y \in \cellset(x_2)}
    \vecnorm{ y - x }$,
where $\cellset(\cdot)$ maps points to their containing cells;
note that $\tilde d$ is everywhere a lower bound for the Euclidean distance.
Noting that the matrix
$\left[ \alpha_{ij}
  = \int_{x_1 \in \cellset^i, x_2 \in \cellset^j}
    \ d\gamma^*(x_1,x_2) \right]_{ij}$
satisfies the constraints of Problem~\ref{problem:traffic_optimistic}, we have
\begin{align*}
  W(\pdistr_\delvtag,\pdistr_\picktag) + \epsilon
  &= \int_{x_1,x_2 \in \env} \vecnorm{ x_2 - x_1 } \ d \gamma^*(x_1,x_2) \\
  &\geq \int_{x_1,x_2 \in \env} \tilde d(x_1,x_2) \ d\gamma^*(x_1,x_2) \\
  &= \sum_{ij}
	\min_{x \in \cellset^i, y \in \cellset^j} \vecnorm{ y-x }
	\int_{ \substack{ x_1 \in \cellset^i \\ x_2 \in \cellset^j } }
	\ d \gamma^*(x_1,x_2) \\
  &\geq \underline\avglink(\cellset).
\end{align*}
Since this inequality holds for \emph{all} $\cellset$ and $\epsilon$ arbitrarily small, we conclude the first part.
Next, we show that for any $\epsilon > 0$,
there exists a partition $\cellset$ so that
$W(\pdistr_\delvtag,\pdistr_\picktag) 
  \leq \underline\avglink(\cellset) + \epsilon$.
For example, we may choose construction $\cellset_\sidecuts$ with resolution $\sidecuts$ sufficiently fine so that
$\vecnorm{x_2-x_1} < \tilde d(x_1,x_2) + \epsilon$ everywhere.
Solve for $\underline\transport(\cellset)$, the solution yielding $\underline\avglink(\cellset)$.
Choose $\hat \gamma \in \Gamma(\pdistr_\delvtag,\pdistr_\picktag)$ to be any distribution that satisfies
$\int_{x_1 \in \cellset^i, x_2 \in \cellset^j} \ d\hat\gamma(x_1,x_2)
  = \underline\alpha_{ij}$ for all $i$ and $j$ (in general, there are infinitely many).
Then we have
$\underline\avglink(\cellset)
  = \int \tilde d(x_1,x_2) \ d\hat\gamma(x_1,x_2)
  \geq \int \vecnorm{x_2-x_1} \ d\hat\gamma(x_1,x_2) - \epsilon
  \geq W(\pdistr_\delvtag,\pdistr_\picktag) - \epsilon$.
\end{proof}

\begin{proof}[Proof of Lemma~\ref{lemma:shadows}]
The delivery sites $Y_1,\ldots,Y_\numdem$ are i.i.d. and independent of $X_1,\ldots,X_\numdem$.
Thus, under the explicitly point-wise independent sampling of Algorithm~\ref{alg:randEBMP}, the $\numdem$ joint variables $(\Yvar_\kvar,J_\kvar,\Xvar'_\kvar)$ for $\kvar=1,\ldots,\numdem$, retain the same set of independencies.
This suffices to prove (i) and the mutual independence of (ii).
To complete the proof, we derive the marginal distribution of some $X'$ explicitly, by writing the joint distribution of $(Y,J,X')$, and then eliminating $(Y,J)$.
The joint distribution can be expressed as
$\pdistr_{YJX'}(y,\jvar,x')
  = \pdistr_\delvtag(y) \ \probtext(\jvar|y) \ \pdistr_{X'|J}(x'|\jvar)
  = \pdistr_\delvtag(y) \left[
      \sum_{\ivar=1}^{\sidecuts^\probdim} \frac{ \alpha_{\ivar\jvar} \indicator{\workcell^\ivar}{y} }{ \pdistr_\delvtag(\workcell^\ivar) } 
      \right]
    \pdistr_\picktag(x'| \Xvar' \in \workcell^\jvar)$;
here $\indicator{\workcell}{\xvar}$ is the indicator function accepting the condition $\xvar\in\workcell$.
Eliminating $Y$, we obtain
$\pdistr_{JX'}(\jvar,\xvar')
  = \int_{\yvar\in\env} \pdistr_{YJX'}(\yvar,\jvar,\xvar') \ d\yvar
  = \int_{y\in\env} \pdistr_\delvtag(y) \left[
      \sum_{\ivar=1}^{\sidecuts^\probdim}
      \frac{ \alpha_{\ivar\jvar} \indicator{\workcell^\ivar}{\yvar} }{ \pdistr_\delvtag(\workcell^\ivar) } 
      \right]
    \pdistr_\picktag(x'| \Xvar'\in\workcell^\jvar) \ dy
  = \pdistr_\picktag(x'| \Xvar'\in\workcell^\jvar) \sum_{\ivar=1}^{\sidecuts^\probdim}
      \alpha_{\ivar\jvar}
  = \pdistr_\picktag(x'|\Xvar'\in\workcell^\jvar) \pdistr_\picktag(\workcell^\jvar)$.
Finally, taking the sum over $j=1,\ldots,\sidecuts^\probdim$, we obtain $\pdistr_{X'}(\xvar') = \pdistr_\picktag(x')$, completing the proof.
\end{proof}

\begin{proof}[Proof of Lemma~\ref{lemma:sampletravel}]
The proof is by derivation:
\begin{equation*}
\begin{split}
\expect \vecnorm{X'-Y}
  &= \int_{\yvar,\xvar'} \vecnorm{\xvar'-\yvar}
      \sum_\jvar \pdistr_{YJX'}(\yvar,\jvar,\xvar') \ d\xvar' \ d\yvar\\
  &= \sum_{\ivar\jvar}
      \frac{ \alpha_{\ivar\jvar} }{ \pdistr_\delvtag(\workcell^\ivar) }
      \int_{\xvar',\yvar \in \workcell^\ivar}
	\vecnorm{\xvar'-\yvar} \pdistr_\delvtag(y)
      \pdistr_\picktag(x'|\Xvar'\in\workcell^\jvar) \ d\yvar \ d\xvar'\\
  &\leq \sum_{\ivar\jvar}
      \frac{ \alpha_{\ivar\jvar} }{ \pdistr_\delvtag(\workcell^\ivar) }
      \max_{\yvar\in\workcell_\ivar,\xvar\in\workcell_\jvar} \vecnorm{\xvar-\yvar}
      \int_{\xvar',\yvar \in \workcell^\ivar} \pdistr_\delvtag(y) \pdistr_\picktag(x'|\Xvar'\in\workcell^\jvar) \ d\yvar \ d\xvar'\\
 & \leq \sum_{\ivar\jvar}
      \alpha_{\ivar\jvar} \max_{\yvar\in\workcell_\ivar,\xvar\in\workcell_\jvar} \vecnorm{\xvar-\yvar}.
      \end{split}
      \end{equation*}
\end{proof}

\begin{proof}[Proof of Lemma~\ref{lemma:diffB}]
Algorithm~\ref{alg:randEBMP2} uses partition $\cellset_\sidecuts$, and optimal solution $\overline\transport(\cellset_\sidecuts)$ of Problem~\ref{problem:traffic_pessimistic},
as inputs to Algorithm~\ref{alg:randEBMP}.
From Lemma~\ref{lemma:sampletravel} we have
$\expect \vecnorm{X'-Y}
  \leq \sum_{\ivar,\jvar}
    \overline\alpha_{\ivar\jvar}
    \max_{\yvar \in \workcell^\ivar, \xvar \in \workcell^\jvar}
    \vecnorm{ \xvar - \yvar }$.
Let $\underline\transport(\cellset_\sidecuts)$ be an optimal solution to Problem~\ref{problem:traffic_optimistic} over partition $\cellset_\sidecuts$, i.e.
$
  \underline\avglink(\cellset_\sidecuts)
  = \sum_{\ivar\jvar}
      \underline\alpha_{\ivar\jvar}
      \min_{y\in\workcell^\ivar,x\in\workcell^\jvar}
      \vecnorm{x-y}$.
Note because $\overline\transport(\cellset_\sidecuts)$ is the optimal solution to Problem~\ref{problem:traffic_pessimistic}, we also have
$\expect \vecnorm{X'-Y}
  \leq \sum_{\ivar\jvar}
    \underline\alpha_{\ivar\jvar}
    \max_{\yvar \in \workcell^\ivar, \xvar \in \workcell^\jvar}
    \vecnorm{ \xvar - \yvar }$.
Finally, we have that
$\underline\avglink(\cellset) \leq W(\pdistr_\delvtag,\pdistr_\picktag)$
for any partition $\cellset$ (see proof of Lemma~\ref{lemma:avgmatch_lowerbound} in Appendix~\ref{appendix:lemmas}).
Combining these results, we obtain
\begin{align*}
  \expect\vecnorm{X'-Y} - W(\pdistr_\delvtag,\pdistr_\picktag)
  & \leq
  \sum_{\ivar\jvar} \underline\alpha_{\ivar\jvar}
  \left[
    \max_{ \substack{ y\in\workcell^\ivar \\ x\in\workcell^\jvar } }
    \vecnorm{x-y}
    -\min_{ \substack{ y\in\workcell^\ivar \\ x\in\workcell^\jvar } }
    \vecnorm{x-y}
  \right]	\\
  & \leq
  ( 2\sidelength\sqrt{\probdim} / \sidecuts )
  \sum_{\ivar\jvar} \underline\alpha_{\ivar\jvar}
  = 2\sidelength\sqrt{\probdim} / \sidecuts.
\end{align*}

\end{proof}

\begin{proof}[Proof of Lemma~\ref{lemma:ebmp_upper}]
We first focus on the case $d\geq 3$.
The proof relies on the characterization of the length of the bipartite matching produced by Algorithm~\ref{alg:randEBMP2} (which also bounds the length of the \emph{optimal} matching).
By the triangle inequality, the length $\tilde \tourlen_\text{M}(Q_\numdem)$ of its matching is at most the sum of the matches between $\pickset$ and $\pickset'$,
plus the distances from the sites in $\delvset$ to their shadows in $\pickset'$, i.e.
\newcommand{\shadowmatch}{\M((\pickset',\pickset))}
\newcommand{\shadowlong}{\tourlen_{\delvset \pickset'}}
\newcommand{\tosubtract}{ {\numdem W(\pdistr_\delvtag,\pdistr_\picktag)} }
\begin{align}\label{eq:fract}
	\tilde \tourlen_\text{M}(Q_\numdem) \leq \shadowmatch + \shadowlong ,
\end{align}
where
$\shadowlong = \sum_{(\Yvar,\Xvar') \in \overline M} \vecnorm{X'-Y}$.
By subtracting on both sides of equation \eqref{eq:fract} the term $\tosubtract$,
and dividing by $\numdem^{1-1/\probdim}$, we obtain
\begin{equation*}
\begin{aligned}
  \frac{
    \tilde\tourlen_\text{M}(Q_\numdem) - \tosubtract
  }{ \numdem^{1-1/\probdim} }
   & \leq
    \frac{ \shadowmatch }{ \numdem^{1-1/\probdim} }
    + \frac{ \shadowlong - \tosubtract 
      }{ \numdem^{1-1/\probdim} }
      \nonumber \\
  & =
    \frac{ \shadowmatch }{ \numdem^{1-1/\probdim} }
    \! +\! \frac{ \shadowlong - \numdem\, \expect\vecnorm{X'-Y} }{ \numdem^{1-1/\probdim} }
      \nonumber 
    \!+\! O\left( \frac{ \numdem^{1/\probdim} }{ \sidecuts } \right) ,
\end{aligned}
\end{equation*}
where the last equality follows from Lemma~\ref{lemma:diffB}.
Lemma~\ref{lemma:shadows} allows us to apply equation~\eqref{eq:EBMP_common}
to $\shadowmatch$, and so the limit of the first term is
\[
\lim_{n\to +\infty} \frac{ \shadowmatch }{ \numdem^{1-1/\probdim} } = \beta_{\text{M},d}\, \int_{\env} \den_{\text{P}}(x)^{1-1/d}\, dx
\]
almost surely.
We observe that
$\numdem \, \expect\vecnorm{X'-Y}$ is the expectation of $\shadowlong$,
and so the second term goes to zero almost surely (absolute differences law, Section~\ref{subsec:absolute_diff}).
The resolution function of Algorithm~\ref{alg:randEBMP2} ensures that the third term vanishes.
Collecting these results, we obtain the inequality in (\ref{eq:BPMdistinct}) with $\phi=\pdistr_\picktag$.
To complete the proof for the case $d\geq 3$, we observe that Algorithm~\ref{alg:randEBMP} could be alternatively defined as follows:
the points in $\pickset$ generate a set $\delvset'$ of shadow sites; the intermediate matching is now between $\delvset$ and $\delvset'$.
One can then prove results congruent with the results in Lemmas~\ref{lemma:shadows}, \ref{lemma:sampletravel}, and~\ref{lemma:diffB}.
By following the same line of reasoning, one can finally prove the inequality in (\ref{eq:BPMdistinct}) with $\phi=\pdistr_\delvtag$. This concludes the proof for the case $d\geq 3$. The proof for the case $d=2$ follows the same logic and is omitted in the interest of brevity.
\end{proof}


\end{document}